\documentclass[lettersize,journal]{IEEEtran}
\usepackage{amsmath,amsfonts}
\usepackage{algorithmic}
\usepackage{algorithm}
\usepackage{array}
\usepackage{textcomp}
\usepackage{stfloats}
\usepackage{url}
\usepackage{verbatim}
\usepackage{graphicx}
\usepackage{cite}
\usepackage{amsmath,amsthm}   
\usepackage{amssymb}
\usepackage{nomencl}
\usepackage{multirow}

\usepackage{bm}
\usepackage{lipsum}
\usepackage{makeidx}
\usepackage{enumerate}
\usepackage{color}

\usepackage[font=small,labelfont=bf]{caption}
\usepackage[font=small,labelfont=bf]{subcaption}

\usepackage{xcolor}
\usepackage{tikz,pgfplots}
\usetikzlibrary{calc}
\usetikzlibrary{intersections}
\usetikzlibrary{arrows,shapes}
\usetikzlibrary{positioning}

\newtheorem{corollary}{Corollary}

\newtheorem{lemma}{Lemma}

\theoremstyle{definition}

\allowdisplaybreaks

\begin{document}
\title{Spectral and Energy Efficiency Maximization of MISO STAR-RIS-assisted URLLC Systems
}
\author{Mohammad Soleymani$^*$, 
Ignacio Santamaria$^\dag$ \emph{Senior Member, IEEE}, and
Eduard Jorswieck$^\ddag$ \emph{Fellow, IEEE}
 \\ \thanks{ 
$^*$Mohammad Soleymani is with the Signal and System Theory Group, Universit\"at Paderborn, Germany, http://sst.upb.de  (email: \protect\url{mohammad.soleymani@sst.upb.de}).  

$^\dag$Ignacio Santamaria is with the Department of Communications Engineering, University of Cantabria (email: \protect\url{i.santamaria@unican.es}).

$^\ddag$ Eduard Jorswieck is with the Institute for Communications Technology, Technische Universit\"at Braunschweig, 38106 Braunschweig, Germany
(e-mail: \protect\url{jorswieck@ifn.ing.tu-bs.de})

The work of I. Santamaria has been partly  supported by the project ADELE PID2019-104958RB-C43, funded by MCIN/AEI/10.13039/501100011033.

The work of Eduard Jorswieck was supported in part by the Federal Ministry of Education and Research (BMBF, Germany) as part of the 6G Research and Innovation Cluster 6G-RIC under Grant 16KISK020K.
}}
\maketitle
\begin{abstract}
This paper proposes a general optimization framework to improve the spectral and energy efficiency (EE) of ultra-reliable low-latency communication (URLLC) simultaneous-transfer-and-receive (STAR) reconfigurable intelligent surface (RIS)-assisted interference-limited systems with finite block length (FBL). This framework can solve a large variety of optimization problems in which the objective and/or constraints are linear functions of the rates and/or EE of users. Additionally, the framework can be applied to any interference-limited system with treating interference as noise as the decoding strategy at receivers. We consider a multi-cell broadcast channel as an example and show how this framework can be specialized to solve the minimum-weighted rate, weighted sum rate, global EE and weighted EE of the system. We make realistic assumptions regarding the (STAR-)RIS by considering three different feasibility sets for the components of either regular RIS or STAR-RIS. Our results show that RIS can substantially increase the spectral and EE of URLLC systems if the reflecting coefficients are properly optimized. Moreover, we consider three different transmission strategies for STAR-RIS as energy splitting (ES), mode switching (MS), and time switching (TS).
We show that STAR-RIS can outperform a regular RIS when the regular RIS cannot  cover all the users.    
Furthermore, it is shown that the ES scheme outperforms the MS and TS schemes.
\end{abstract} 
\begin{IEEEkeywords}
 Energy efficiency, finite block length, majorization minimization, MISO broadcast channels, reflecting intelligent surface, spectral efficiency, ultra-reliable low-latency communications.
\end{IEEEkeywords}

\section{Introduction}
The sixth generation (6G) of wireless systems should be able to support many different applications such as ultra-reliable low-latency communication (URLLC), massive machine-type communication (mMTC), enhanced mobile broadband (eMBB) and internet of things (IoT) \cite{vaezi2022cellular, chafii2022ten, almekhlafi2022enabling, durisi2016toward, popovski2019wireless}. 
These applications typically require very low decoding error probabilities as well as very low latency and enforce us to employ a packet size much shorter than human type communications \cite{bockelmann2016massive}. Additionally, there may exist tens of billions of  various machine-type
terminals such as sensors, vehicles, drones, and robots in mMTC and/or IoT networks \cite{vaezi2022cellular, bockelmann2016massive}.  To fulfill these demands, spectral and energy efficiency (EE) should be drastically improved \cite{vaezi2022cellular}. 
Recently, it has been shown that reconfigurable intelligent surface (RIS) can be a promising technology for enhancing the performance of various wireless networks \cite{wu2021intelligent, huang2020holographic, di2020smart}.  

In this work, we investigate the performance of RIS (either regular or simultaneous transmit and reflect (STAR)) in URLLC systems with finite block length (FBL) and propose a unified optimization framework to improve the spectral and EE of (STAR-)RIS-assisted URLLC systems. 
\subsection{Literature Review}
Modern mobile communication systems are expected to accomplish many wide applications and services in various areas such as  autonomous driving, healthcare, augmented reality, automation in industry, and so on \cite{simsek20165g, aceto2019survey,9815179}. 
As indicated, these application may require FBL in which the Shannon rate may not be  
achievable \cite{polyanskiy2010channel}. 
In this case, the achievable rate for a point-to-point system with Gaussian independent and identically distributed (iid) signals can be approximated as
\cite{polyanskiy2010channel}
\begin{equation}\label{na}
{r}=
C
-
\alpha\sqrt{{V}},
\end{equation}
where $C$ is the Shannon Rate, $\alpha$ is a constant value, which is a function of the packet length and the desired decoding error probability, and $V$ is the channel dispersion. 
The rate approximation in \eqref{na} is widely known as the normal approximation (NA) \cite{polyanskiy2010channel}. The NA is known to be accurate for packet lengths more than around $124$ bits and decoding error probabilities higher than $10^{-5}$ \cite{erseghe2016coding, erseghe2015evaluation, cocskun2019efficient, lancho2019single }. 
Note that in the third generation partnership project (3 GPP),
the packet length is $32$ bytes (or equivalently $256$ bits), and the error probability is $\epsilon=10^{-5}$  \cite[Sec. II.D]{popovski2019wireless}.
In general,  
 the reliability constraint and/or packet length may vary for different services \cite{popovski20185g}. For instance, according to \cite{popovski20185g}, the packet error probability should be in order of $10^{-5}$ for URLLC, $10^{-3}$ for eMBB, and $10^{-1}$ for mMTC. It appears that the NA should be accurate enough for these parameters if we operate in  moderate or high SNR regimes \cite{erseghe2016coding, erseghe2015evaluation}.

To be able to support URLLC, one of the targets of 6G is to improve the spectral and EE \cite{vaezi2022cellular}. Note that the EE of a system is defined as the ratio between the system throughput and the total power consumption \cite{zappone2015energy}. Energy efficient techniques are very important for IoT and/or industrial IoT (IIoT) and/or MTC since such techniques can increase the battery life of devices and reduce the implementation and/or maintenance expenses \cite{vaezi2022cellular, aceto2019survey}. 
The spectral efficiency and/or EE of different systems with FBL has been studied in \cite{schiessl2020noma,  schiessl2019delay, schiessl2018delay, ghanem2020resource, nasir2020resource, nasir2021cell}. 
The authors in \cite{schiessl2020noma} considered the performance of non-orthogonal multiple-access (NOMA) in a multiple-access channel (MAC) with FBL. The paper \cite{schiessl2019delay} studied the delay performance of multiple-input, single-output (MISO) broadcast channel (BC) with imperfect channel state information (CSI) and FBL. The authors in \cite{ghanem2020resource} maximized the sum-rate of a single-cell MISO BC with orthogonal-frequency-division-multiple-access (OFDMA)-URLLC. 
The paper\cite{nasir2020resource} maximized the minimum rate of a MISO URLLC BC. The authors in \cite{nasir2021cell} studied a massive multiple-input, multiple-output (MIMO) URLLC with FBL and proposed schemes to maximize the minimum rate and EE of the network.

To improve spectral and EE, 6G will employ some emerging technologies such as RIS, which can enhance the coverage and consequently, the system performance \cite{wu2021intelligent, di2020smart, huang2020holographic, elmossallamy2020reconfigurable}. RIS has been employed to improve the performance of various systems with realistic assumptions regarding the CSI and devices \cite{wu2019intelligent, yang2020risofdm, huang2019reconfigurable,  kammoun2020asymptotic, yu2020joint, soleymani2023rate,  pan2020multicell, zhang2020intelligent, santamaria2023interference, ni2021resource, soleymani2022improper, soleymani2022noma, soleymani2022rate}.  
For instance, in \cite{pan2020multicell, soleymani2022improper}, a multi-cell MIMO RIS-assisted BC was considered and it was shown that RIS can improve the spectral and EE of the system. 
The superiority of RIS in a single-cell MISO BC was shown in \cite{huang2019reconfigurable,kammoun2020asymptotic, yu2020joint}. 
In \cite{ni2021resource, soleymani2022noma},  NOMA was applied to multi-cell RIS-assisted BCs and it was shown that NOMA can enhance the system performance. 
The authors in \cite{yang2020risofdm} showed that RIS can improve the performance of an OFDM system. We refer the reader to \cite{wu2021intelligent, di2020smart} for an overview of RIS.

A regular RIS can only reflect signals, which may restrict the coverage area of the RIS. For a $360^\circ$ coverage, a STAR-RIS has been proposed
in which each passive element can reflect and transmit at the same time \cite{mu2021simultaneously, wu2021coverage, liu2021star, xu2021star,9774942,ni2022star, xie2022star, zhang2022secrecy, nguyen2022performance}. 
STAR-RIS is a novel technology and has been evaluated by experimental results \cite{docomo2020docomo}.
Due to a wider coverage area by STAR-RIS, it can be expected that STAR-RIS is able to support more applications especially when it is not possible to locate a regular RIS such that all transceivers are in its reflection area. 

All the aforementioned papers on RIS, \cite{wu2021intelligent, di2020smart, huang2020holographic, elmossallamy2020reconfigurable, huang2019reconfigurable, wu2019intelligent, kammoun2020asymptotic, yu2020joint,  yang2020risofdm, pan2020multicell, zhang2020intelligent,  ni2021resource, soleymani2022improper, soleymani2022noma, soleymani2022rate}, considered the performance of RIS in systems with infinite block length. The performance of  RIS in the presence of FBL has been studied in
\cite{  
li2021aerial,  vu2022intelligent, ren2021intelligent, 
xie2021user,  zhang2021irs,   almekhlafi2021joint, 
ranjha2020urllc, ranjha2021urllc, 
 dhok2021non, ghanem2021joint, abughalwa2022finite}. 
The paper \cite{li2021aerial} studied the performance of RIS and unmanned aerial vehicles (UAVs) in URLLC systems and showed that RIS can be beneficial in the system. 
The authors in \cite{vu2022intelligent} considered NOMA in RIS-assisted single-input, single-output (SISO) BC with two users and showed that RIS can improve the system throughput and decrease the average decoding error rate. 
 The performance of RIS in MISO point-to-point URLLC systems with a factory automation scenario was investigated in \cite{ren2021intelligent}, where  the average data rate and decoding error probability were obtained under different assumptions regarding the fading and propagation of the channels. 
The authors in \cite{xie2021user} minimized the latency of a single-cell SISO RIS-assisted URLLC BC with user grouping.
The paper \cite{zhang2021irs} considered RIS-assisted URLLC systems with FBL to transfer information and energy wirelessly. 
The paper \cite{almekhlafi2021joint} proposed resource allocation schemes for RIS-assisted single-cell BCs with the coexistence of eMBB and URLLC services.

\begin{table*}
\centering
\scriptsize
\caption{A brief comparison of the most related works.}\label{table-1}
\begin{tabular}{|c|c|c|c|c|c|c|c|c|c|c|c|c|c|c|}
	\hline
&RIS&STAR-RIS&URLLC&Multi-user communications&Multiple-antenna Systems&Channel dispersion in \cite{scarlett2016dispersion}& EE metrics
 \\
\hline
  This paper&$\surd$&$\surd$&$\surd$&$\surd$&$\surd$&$\surd$&$\surd$
\\
\hline
\cite{ nasir2020resource,  ghanem2020resource}&&&$\surd$&$\surd$&$\surd$&&
\\
\hline
\cite{schiessl2020noma}&&&$\surd$&$\surd$&&$\surd$&
\\
\hline
\cite{schiessl2019delay}&&&$\surd$&$\surd$&$\surd$&$\surd$&
\\
\hline
\cite{nasir2021cell}&&&$\surd$&$\surd$&$\surd$&$\surd$&$\surd$
\\
\hline
\cite{ kammoun2020asymptotic, yu2020joint,   pan2020multicell, zhang2020intelligent}&$\surd$&&&$\surd$&$\surd$&&
\\
\hline
\cite{huang2019reconfigurable, soleymani2022improper, soleymani2022noma}&$\surd$&&&$\surd$&$\surd$&&$\surd$
\\
\hline
\cite{soleymani2022rate}&$\surd$&$\surd$&&$\surd$&$\surd$&&$\surd$
\\
\hline
\cite{ li2021aerial, ghanem2021joint }&$\surd$&&$\surd$&$\surd$&$\surd$&&
\\
\hline
\cite{ vu2022intelligent, xie2021user,     almekhlafi2021joint}&$\surd$&&$\surd$&$\surd$&&&
\\
\hline
\cite{ ren2021intelligent,   zhang2021irs}&$\surd$&&$\surd$&&$\surd$&&
\\
\hline
\cite{abughalwa2022finite}&$\surd$&&$\surd$&$\surd$&&$\surd$&
\\
\hline
		\end{tabular}
\normalsize
\end{table*}

It is worth emphasizing that the optimal channel dispersion in \eqref{na} is 
\begin{equation}
V^{opt}=1-\frac{1}{\left(1+\gamma\right)^2},
\end{equation}
where $\gamma$ is the SNR. 
Even if we treat interference as noise,  we cannot simply replace the SNR term by signal-to-interference-plus-noise ratio (SINR) and use the same channel dispersion for interference-limited systems \cite{scarlett2016dispersion}. Indeed the optimal channel dispersion cannot be achieved by Gaussian signals in the presence of interference. Unfortunately, this issue is sometimes overlooked in the literature when studying an interference-limited system. In this paper, we consider the channel dispersion in \cite{scarlett2016dispersion}, which can be achieved by Gaussian signals in interference-limited systems with treating interference as noise (TIN). To the best of our knowledge, \cite{abughalwa2022finite} is the only work in interference-limited RIS-assisted systems with FBL that considered the channel dispersion in \cite{scarlett2016dispersion}.
Note that \cite{abughalwa2022finite} studied the geometric mean of the rates. However, in this paper, we consider various utility functions as will be discussed in the next subsection.

\subsection{Motivations and contributions}
The main motivation for this work is to propose a general framework for URLLC RIS-assisted systems. In Table \ref{table-1}, we provide a brief summary of the most related works based on the considered scenario and the channel dispersion.  As can be observed, even though RIS has received many attentions during recent years, the performance of RIS in URLLC with FBL should be further investigated, especially in multi-user communication  systems. 
For instance, to the best of our knowledge, there is no work on STAR-RIS in URLLC with FBL. Additionally, EE metrics have not been considered in RIS-assisted URLLC systems with FBL. 
It should be emphasized that the rate expressions with FBL are more complicated than the Shannon rates, and it is not straightforward to modify and apply existing optimization approaches to URLLC RIS-assisted systems with FBL. Thus, a unified optimization framework is required to study the performance of such systems by considering STAR-RIS and EE metrics.  

In this paper, we propose a general optimization framework for URLLC with FBL in STAR-RIS-assisted systems. To the best of our knowledge, this is the first work to study the performance of STAR-RIS in URLLC systems. Our proposed framework can be applied to any optimization problem in which the objective and/or the constraints are linear functions of the rates and/or EE of users. 
Moreover, the framework can be applied to any interference-limited system, using treating interference as noise (TIN) as the decoding strategy.

We consider a multi-cell MISO RIS-assisted BC as an illustrative example and show how the unified framework can be specialized to solve different optimization problems. To this end, we consider the minimum-weighted-rate, weighted-sum-rate, global EE and minimum-weighted-EE maximization problems. We refer to the objective function of these optimization problems as the utility function.

We make realistic assumptions regarding RIS (either regular or STAR) by modeling the small-scale and large-scale fading and considering three different feasibility sets based on the models in \cite{wu2021intelligent,xu2021star,wu2021coverage,liu2021star,9774942}. 
We propose three main approaches to optimize reflecting/transmitting coefficients in STAR-RIS. First, we assume that a set of RIS components operate only in the reflection mode, and the remaining components operate only in the transmission mode. This scheme is referred to as the mode switching \cite{mu2021simultaneously}. Second, we assume that all the components simultaneously operate in both reflection and transmission mode, which is referred to as the energy splitting mode \cite{mu2021simultaneously}. Third, we assume that each time slot is divided into two sub-slots. Then, all the components operate in the reflection mode in the first sub-slot, while they all operate in the transmission mode in the next sub-slot. This scheme is referred to as time switching \cite{mu2021simultaneously}. 

Through numerical examples, we show that RIS can significantly improve the spectral or the energy efficiency of the system when the reflecting coefficients are properly optimized. Interestingly, it may happen in some special cases that RIS worsen the performance if the reflecting coefficients are chosen randomly. Additionally, we show that STAR-RIS with mode switching and energy splitting approaches outperforms regular RIS when the RIS cannot cover all the users.  The mode-switching scheme  of STAR-RIS slightly improves the system performance over regular RIS. However, the energy-splitting scheme considerably outperforms the regular RIS. 

The main contributions of this work can be summarized as follows:
\begin{itemize} 
\item We propose a unified optimization framework to solve a large family of optimization problems for MISO STAR-RIS-assisted interference-limited URLLC systems. This framework can be applied to any interference-limited system with TIN. 

\item We consider a multicell BC and specialize the framework to solve minimum-weighted rate, weighted-sum rate,  minimum-weighted rate, and global EE maximization problems.

\item We study the performance of STAR-RIS with three different feasibility sets and consider three schemes for optimizing the reflecting/transmitting coefficients of STAR-RIS. We show that STAR-RIS may outperform regular RIS if the regular RIS cannot cover all the users. 

\item We show that RIS can substantially improve the spectral and EE of RIS-assisted  URLLC systems by considering different utility functions and feasibility sets for RIS components. 
\end{itemize}
\subsection{Paper outline}
The rest of the paper is organized as follows. Section \ref{sec-sym} presents the system model and formulate the considered problem. 
Section \ref{sec-iii} presents the generalized optimization framework for the systems without RIS. 
Section \ref{seciv} states the extension of the proposed optimization framework to URLLC (STAR-)RIS-assisted systems. 
Section \ref{num-sec} provides some numerical results. Section \ref{con-sec} concludes the paper. Finally, we provide some proofs in appendices.  

\begin{table}
\centering
\footnotesize
\caption{List of frequently used notations.}\label{table-2}
\begin{tabular}{|l|l|}
	\hline
$L/M$ & No. of BSs/RISs \\
$K$ & No. of users per each cell\\
$N_{BS}/N_{RIS}$
& No. of antennas/components at each BS/RIS\\
u$_{lk}$& $k$-th associated user to  BS $l$\\
${s}_{lk}$
& Transmit  signal of BS $l$ intended for u$_{lk}$\\
$\mathbf{x}_{lk}$
& Beamforming vector of BS $l$ corresponding to ${s}_{lk}$\\
$\mathbf{h}_{lk,i}$
& Equivalent channel between BS $i$ and u$_{lk}$\\
$\mathbf{d}_{jk,l}$
 & Direct channel between the BS $l$ and u$_{jk}$\\  
$\mathbf{f}_{jk,m}$
&Channel between the $m$th RIS and u$_{jk}$\\ 
$\mathbf{G}_{ml}$
&Channel between the BS $l$ and the $m$th RIS\\
$\bm{\Theta}_m$
& Matrix of RIS components for the $m$th RIS \\
$\bm{\Theta}_m^{t/r}$
& Matrix of STAR-RIS components for the TS/RS at  RIS $m$\\
$\epsilon^c$&Decoding error probability\\
$\mathcal{X}$&The feasibility set for beamforming vectors
\\
$\mathcal{T}$&The feasibility set for (STAR-)RIS components
\\
${n}_{lk}$
&Additive white Gaussian noise at u$_{lk}$\\
$e_{lk}/r_{lk}$&Energy efficiency/rate of u$_{lk}$\\
$V_{lk}$&Channel dispersion at u$_{lk}$\\
$p_l$&Power budget of BS $l$\\
$\gamma_{lk}$&SINR at u$_{lk}$\\
$n_t$&Packet length\\
\hline
		\end{tabular}
\normalsize
\end{table} 

\begin{figure}[t!]
    \centering
\includegraphics[width=.45\textwidth]{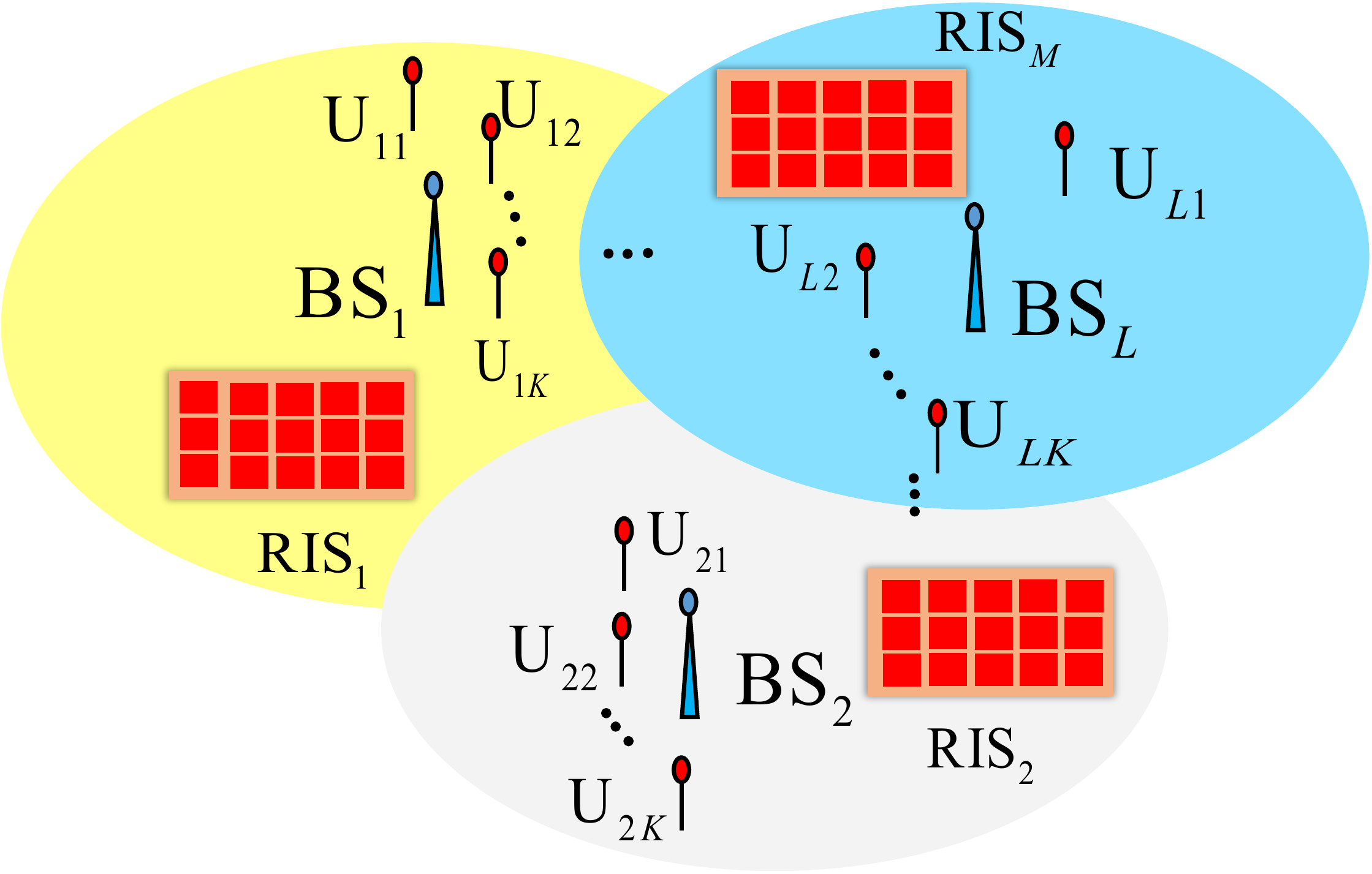}
     \caption{A multicell broadcast channel with RIS.}
	\label{Fig-sys-model}
\end{figure}
\section{System model}\label{sec-sym}
Our unified optimization framework can be applied to a large class of interference-free and/or interference-limited URLLC MISO (STAR)-RIS-assisted systems with TIN. Such systems include, for instance, various multi-user interference channels, cognitive radio systems, broadcast channels, multiple-access channels, device-to-device communications, and so on.
For the sake of illustration, we consider a multicell broadcast channel with $L$ multi-antenna base stations (BSs) in which each BS has $N_{BS}$ antennas and serves $K$ single-antenna users. 
We assume that there are $M\geq L$ RISs with $N_{RIS}$ reflecting elements in the system to assist the BSs. 
Note that a multicell BC is a practical scenario, which is considered in many work such as \cite{xie2022star, soleymani2022rate, ni2021resource, pan2020multicell}. In a multicell BC, intercell interference may highly degrade the system performance especially for cell-edge users, which should be handled by a joint optimization of transmit parameters at BSs. 

We assume perfect, global and instantaneous channel state information (CSI) at all transceivers similar to  many other works on RIS (either STAR or regular) such as \cite{ni2021resource, huang2019reconfigurable, wu2019intelligent, kammoun2020asymptotic, pan2020multicell, zhang2020intelligent, zuo2020resource, mu2020exploiting, yang2021reconfigurable, yu2020joint, jiang2022interference, liu2022simultaneously, 9774942}. 
Additionally, it should be noted that in this work, we focus on resource allocation for URLLC systems in which it is common to assume  perfect, global and instantaneous CSI \cite{nasir2020resource, wang2023flexible, singh2020energy, ghanem2021joint, he2021beamforming, choi2021mimo, ghanem2020resource}.
Investigating the performance of RIS with perfect CSI can show the main tradeoffs in the system design and provide an upper bound for the system performance. Indeed, studying the performance of RIS with perfect CSI is useful to show whether/how (STAR-)RIS provides any benefit in URLLC systems. 
If the benefits of RIS are minor with perfect CSI, then it may suggest that RIS cannot be beneficial in more realistic scenarios. 
We also assume that BSs use short-length packets to transmit data to the users. 
Without loss of generality, we consider a symmetric scenario with the same number of BS antennas or users per cell. 
However, our work can be easily extended to an asymmetric scenario in which each BS has a different number of antennas/associated users. 
\begin{figure}[t!]
    \centering
\includegraphics[width=.3\textwidth]{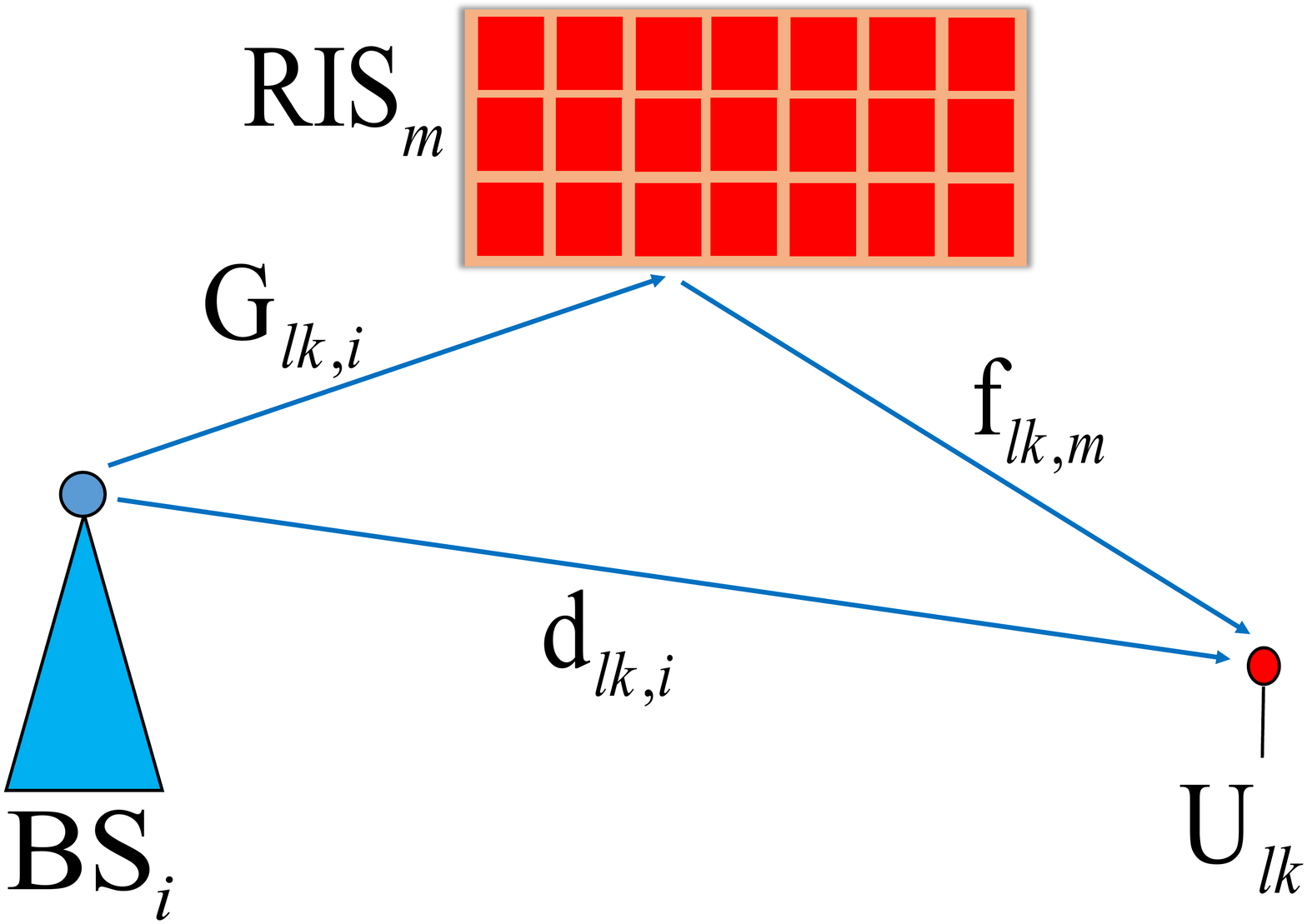}
     \caption{Channel model in a typical RIS-assisted system.}
	\label{Fig-sys-model2}
\end{figure}

\subsection{RIS model}\label{sec-ris}
 In this paper, we consider both regular and STAR-RISs. Here, we 
briefly describe the model here and refer the readers to \cite{pan2020multicell, soleymani2022improper}, \cite[Sec. II]{wu2021intelligent} or \cite{9774942} for a detailed description/review on the features of RIS and/or STAR-RIS.
\subsubsection{Regular RIS}As shown in Fig. \ref{Fig-sys-model2}, the channel between BS $i$ and user $k$ associated to BS $l$, i.e., u$_{lk}$, is 
 \begin{equation}\label{eq-ch}
\mathbf{h}_{lk,i}\left(\{\bm{\Theta}\}\right)=\underbrace{\sum_{m=1}^M\mathbf{f}_{lk,m}\bm{\Theta}_m\mathbf{G}_{mi}}_{\text{Links through RIS}}+
\underbrace{\mathbf{d}_{lk,i}}_{\text{Direct link}}\in \mathbb{C}^{1\times N_{BS}},
\end{equation}
where $\mathbf{d}_{lk,i}$ 
is the direct link between BS $i$ and u$_{lk}$, $\mathbf{G}_{mi}$ 
is the channel matrix between BS $i$ and RIS $m$, $\mathbf{f}_{lk,m}$ 
is the channel vector between RIS $m$ and u$_{lk}$, 
and  $\bm{\Theta}_m$ 
is
\begin{equation}
\bm{\Theta}_m=\text{diag}\left(\theta_{m1}, \theta_{m2},\cdots,\theta_{mN_{RIS}}\right),
\end{equation}
where $\theta_{mi}$s for all $m$, $i$ are the reflecting coefficients. 
Hereafter, we drop the dependency of the channels with respect to $\{\bm{\Theta}\}$ for notational simplicity. 

The channels can be optimized only through  the reflecting coefficients $\theta_{mi}$, which are complex-valued parameters. There are different assumptions for modeling the reflecting coefficients. In this paper, we consider three different feasibility sets for $\theta_{mi}$s based on the models in \cite[Sec. II]{wu2021intelligent}. 
An upper bound for RIS performance can be achieved by considering the amplitude and phase of  $\theta_{mi}$s as independent random variables, which results in the following feasibility set \cite[Eq. (11)]{wu2021intelligent}
\begin{equation}
\mathcal{T}_{U}=\left\{\theta_{m_n}:|\theta_{m_n}|^2\leq 1 \,\,\,\forall m,n\right\}.
\end{equation}
 A more common feasibility set is 
\begin{equation}
\mathcal{T}_{I}=\left\{\theta_{mi}:|\theta_{mi}|= 1 \,\,\,\forall m,i\right\},
\end{equation}
which  has been widely used in the literature \cite{di2020smart, wu2021intelligent, wu2019intelligent, kammoun2020asymptotic, yu2020joint, pan2020multicell, zhang2020intelligent}.
Another practical feasibility set is \cite{abeywickrama2020intelligent}
\begin{equation}
\mathcal{T}_{C}\!=\!\left\{\theta_{mi}\!:|\theta_{mi}|= \mathcal{F}(\angle \theta_{mi}), \,\angle \theta_{mi}\in[-\pi,\pi]  \,\forall m,i\right\}\!.
\end{equation}
In this model, the amplitude  of each RIS element is a deterministic function of its phase as \cite{abeywickrama2020intelligent}
\begin{equation}\label{eq*=*}
\mathcal{F}(\angle \theta_{mi})= |\theta|_{\min}+( 1\!-|\theta|_{\min})\left(\!\!\frac{\sin\left(\angle \theta_{mi}-\phi\right)+1}{2}\!\right)^{\alpha}\!\!\!,
\end{equation}
where $|\theta|_{\min}$, $\alpha$, and $\phi$ are non-negative constant values. 
\begin{figure}[t!]
    \centering
\includegraphics[width=.45\textwidth]{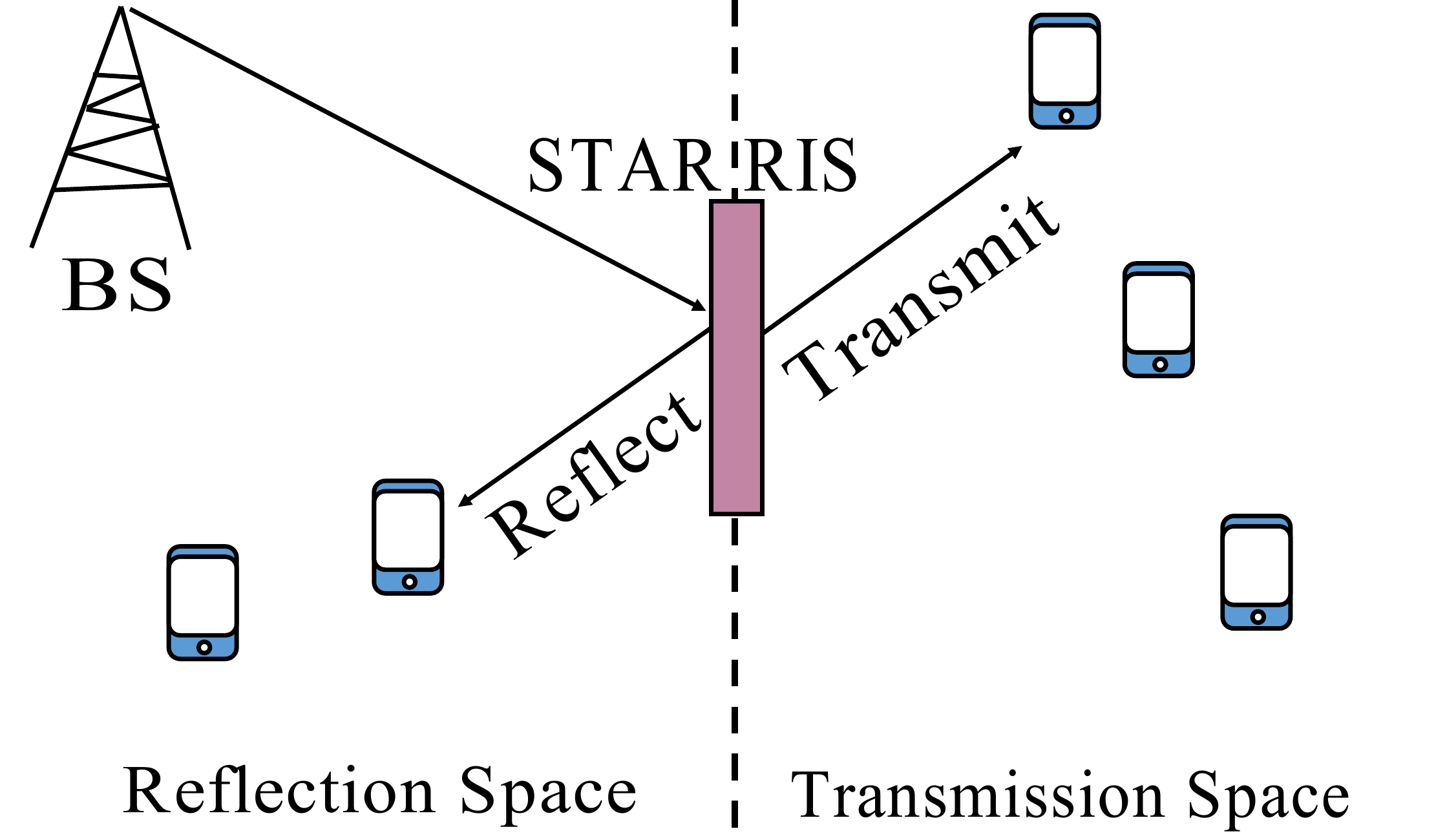}
     \caption{A typical STAR-RIS-assisted system.}
	\label{Fig-3-sys}
\end{figure}
\subsubsection{STAR-RIS}\label{sec=ap=ris-d}
In a regular RIS, each RIS component can only reflect the signals. However, STAR-RIS allows each component to not only reflect, but also transmit signals simultaneously, as its name suggests. Thus, a STAR-RIS can cover a larger area as shown in Fig. \ref{Fig-3-sys}. Indeed, there are two spaces for each RIS: reflection space (RS) and transmission space (TS), while a regular RIS can only reflect \cite{mu2021simultaneously}. 
Note that in this paper, we consider only passive STAR-RIS in which each STAR-RIS element cannot amplify the signal power neither in the RS nor in TS. 

In STAR-RIS-assisted systems,  each user belongs to either RS or TS \cite{mu2021simultaneously}.
We represent the reflecting and transmit coefficients of the $i$-th component of the $m$-th RIS by, respectively, $\theta_{mi}^{r}$ and $\theta_{mi}^{t}$.
In case of STAR-RIS, the channel between BS $i$ and u$_{lk}$ is  \cite[Eq. (2)]{mu2021simultaneously} 
\begin{equation*}
\mathbf{h}_{lk,i}
=\underbrace{\sum_{m=1}^M\mathbf{f}_{lk,m}\bm{\Theta}_m^{r/t}\mathbf{G}_{mi}}_{\text{Links through RIS}}+
\underbrace{\mathbf{d}_{lk,i}}_{\text{Direct link}},
\end{equation*}
where
$\bm{\Theta}_m^r
=\text{diag}\left(\theta_{m_1}^r, \theta_{m_2}^r,\cdots,\theta_{m_{N_{RIS}}}^r\right)$, and 
$\bm{\Theta}_m^t
=\text{diag}\left(\theta_{m_1}^t, \theta_{m_2}^t,\cdots,\theta_{m_{N_{RIS}}}^t\right)$. 
The upper bound for the performance of STAR-RIS is given by the following feasibility set \cite[Eq. (2)]{xu2021star}
\begin{equation}\label{model-2}
\mathcal{T}_{SU}=\left\{\theta_{mn}^r,\theta_{mn}^t:|\theta_{m_i}^{r}|^2+|\theta_{m_i}^{t}|^2\leq 1 \,\,\,\forall m,n\right\}.
\end{equation}
 A more common feasibility set is \cite{wu2021coverage}, \cite[Eq. (1)]{liu2021star}
\begin{equation}\label{model-1}
\mathcal{T}_{SI}=\left\{\theta_{mn}^r,\theta_{mn}^t:|\theta_{m_i}^{r}|^2+|\theta_{m_i}^{t}|^2= 1 \,\,\,\forall m,i\right\}.
\end{equation}
In these feasibility sets, the phases of the reflection and transmission coefficients can be independently optimized.  
However, another feasibility set has been studied in \cite{9774942, liu2022simultaneously} in which the phases completely depend on each other. 
This feasibility set can be written as \cite[Proposition 1]{9774942}
\begin{multline}\label{model-3}
\mathcal{T}_{SN}=\left\{\theta_{mn}^r,\theta_{mn}^t:|\theta_{m_i}^{r}|^2+|\theta_{m_i}^{t}|^2= 1,\right.
\\
\left. \mathfrak{R}\left\{\theta_{m_i}^{r^*}\theta_{m_i}^{t}\right\}= 0\,\,\,\forall m,i\right\}.
\end{multline}
\begin{lemma}
The two constraints $|\theta_{m_i}^{r}|^2+|\theta_{m_i}^{t}|^2= 1$ and $\mathfrak{R}\left\{\theta_{m_i}^{r^*}\theta_{m_i}^{t}\right\}= 0$ are equivalent to the 
following constraints:
\begin{align}\label{st-eq-}
|\theta_{m_i}^{r}+\theta_{m_i}^{t}|^2&\leq 1,\\
\label{st-eq-1}
|\theta_{m_i}^{r}-\theta_{m_i}^{t}|^2&\leq 1,\\
|\theta_{m_i}^{r}|^2+|\theta_{m_i}^{t}|^2&= 1.
\label{st-eq-2}
\end{align}
\end{lemma}
\begin{proof}
We have the following equality 
\begin{multline}
|\theta_{m_i}^{r}\pm\theta_{m_i}^{t}|^2=|\theta_{m_i}^{r}|^2+|\theta_{m_i}^{t}|^2\pm2\mathfrak{R}\left\{\theta_{m_i}^{r^*}\theta_{m_i}^{t}\right\}.
\end{multline}
Substituting $|\theta_{m_i}^{r}|^2+|\theta_{m_i}^{t}|^2$ by $1$, the constraints \eqref{st-eq-} and \eqref{st-eq-1} simplify to
$\pm\mathfrak{R}\left\{\theta_{m_i}^{r^*}\theta_{m_i}^{t}\right\}\leq 0$, 
which is equivalent to $\mathfrak{R}\left\{\theta_{m_i}^{r^*}\theta_{m_i}^{t}\right\}=0$.
\end{proof}

\subsection{Signal model}
The broadcast signal from BS $l$ is 
\begin{equation}
\mathbf{x}_l=
\sum_{k=1}^K\mathbf{x}_{lk}s_{lk}
\in\mathbb{C}^{N_{BS}\times 1},
\end{equation}
 where $s_{lk}\sim\mathcal{CN}\left(0,1\right)$ is the  message intended for the $k$-th user associated to the $l$-th BS, denoted by u$_{lk}$, and $\mathbf{x}_{lk}$ is the corresponding beamforming vectors, which is an optimization parameter. 
Note that $s_{lk}$s for all $l,k$ are iid proper Gaussian signals.

The received signal for the user u$_{lk}$ is
\begin{subequations}\label{rec-sig}
\begin{align}
y_{lk}&=\sum_{i=1}^L\mathbf{h}_{lk,i}\sum_{j=1}^K\mathbf{x}_{ij}s_{ij}
+n_{lk}
\\&
=
\underbrace{\mathbf{h}_{lk,l}
\mathbf{x}_{lk}s_{lk}}_{\text{Desired signal}}+
\underbrace{\mathbf{h}_{lk,l}\!\!\!\!
\sum_{j=1,j\neq k}^{K}\!\!\!\!\mathbf{x}_{lj}
s_{lj}}_{\text{Intracell interference}}+
\underbrace{ \!\!\sum_{i=1,i\neq l}^L\!\!\!\mathbf{h}_{lk,i}
\mathbf{x}_{i}}_{\text{Intercell interference}}\!\!\!+
\underbrace{n_{lk}}_{\text{Noise}},
\end{align}
\end{subequations}
where 
$\mathbf{h}_{lk,l}\in\mathbb{C}^{1\times N_{BS}}$ is the  channel between BS $l$ and u$_{lk}$, and $n_{lk}\sim\mathcal{CN}\left(0,\sigma^2\right)$ is additive  white Gaussian noise. 
As can be easily verified through \eqref{rec-sig}, the received signal and the interference term are proper Gaussian signals. 
Furthermore, the noise, interference and desired signals are independent from each other. 
Note that there exists two types of interference in the received signal $y_{lk}$: intracell and intercell interference. Intracell interference at u$_{lk}$ is caused by BS $l$, while intercell interference is caused by the other BSs.

\subsection{Rate and energy efficiency expressions}\label{sec=re}
Each user decodes its own message, treating all other signals as noise. 
Thus, the rate of  u$_{lk}$ is  \cite{polyanskiy2010channel}, \cite[Eq. (8)]{nasir2021cell}
\begin{equation}\label{rate}
{r}_{lk}=
\underbrace{\log\left(1+\gamma_{lk}\right)}_{\text{\small Shannon Rate}}
-
\underbrace{Q^{-1}(\epsilon^c)\sqrt{\frac{V_{lk}}{n_t}}}_{\delta_{lk}(\{\mathbf{x}\},\{\bm{\Theta}\})},
\end{equation}
where $V_{lk}$ is the channel dispersion for decoding $s_{lk}$ at u$_{lk}$, $n_t$ is the packet length, $Q^{-1}$ is the inverse of the Gaussian Q-function, $\epsilon^c$ is an acceptable decoding error probability, which indicates that one out of $1/\epsilon^c$ short-length packets may experience outage, 
and  $\gamma_{lk}$ is the corresponding SINR given by
\begin{equation}\label{sinr}
\gamma_{lk}=
\frac{|\mathbf{h}_{lk,l}
\mathbf{x}_{lk}|^2}
{\sigma^2
+
\sum_{ij\neq lk}|\mathbf{h}_{lk,i}
\mathbf{x}_{ij}|^2},
\end{equation}
where 
\begin{equation}%
\sum_{[ij]\neq [lk]}|\mathbf{h}_{lk,i}
\mathbf{x}_{ij}|^2=\sum_{ij}|\mathbf{h}_{lk,i}
\mathbf{x}_{ij}|^2-|\mathbf{h}_{lk,l}
\mathbf{x}_{lk}|^2. 
\end{equation}
We represent the gap between the Shannon rate and the FBL by $\delta_{lk}(\{\mathbf{x}\},\{\bm{\Theta}\},\epsilon^c)=Q^{-1}(\epsilon^c)\sqrt{
{V_{lk}}/{n_t}}$, which is a function of the channel dispersion. 
Note that $\epsilon^c$ is related to the reliability constraint, and the gap between the Shannon rate and the FBL increases with the reliability. In other words, to ensure a more reliable communication (lower $\epsilon^c$), we should transmit data at a rate lower than the Shannon rate. 
The optimal channel dispersion is \cite{polyanskiy2010channel}
\begin{equation}\label{dis-opt}
V_{lk}^{opt}=1-\frac{1}{\left(1+\gamma_{lk}\right)^2}=\frac{\gamma_{lk}}{1+\gamma_{lk}}\left(1+\frac{1}{1+\gamma_{lk}}\right).
\end{equation}
However, it is not achievable by iid Gaussian signals when there exists interference. 
In \cite{scarlett2016dispersion}, the authors proposed a simple and practical scheme, which is not dispersion optimal. 
The channel dispersion for the scheme in \cite{scarlett2016dispersion} is
\begin{equation}\label{dis-}
V_{lk}=2\frac{\gamma_{lk}}{1+\gamma_{lk}}=2\left(1-\frac{1}{1+\gamma_{lk}}\right).
\end{equation}
Note that it can be easily verified that the dispersion in \eqref{dis-} is an upper bound for  \eqref{dis-opt}. 
Additionally, it should be noted that the FBL rate, Shannon rate, channel dispersion and $\delta$ are a function of $\epsilon^c$, $\{\mathbf{x}\}$ and $\bm{\Theta}$. However, due to a notational simplicity, we drop this dependency in the equations unless it causes confusion. 
\begin{figure}[t!]
    \centering
\includegraphics[width=.49\textwidth]{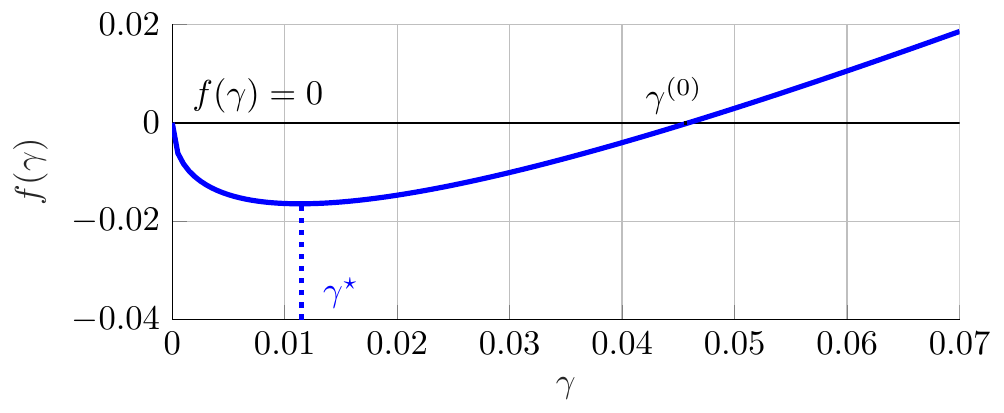}
    \caption{$f(\gamma)$ versus $\gamma$ for $a=0.2185$. It represents $r_{l,k}/\ln2$ as a function of $\gamma_{l,k}$ for $n_t=200$, and $\epsilon=10^{-3}$.
}
	\label{Fig-3} 
\end{figure}
In the following lemma, we take a look at the structure of \eqref{rate}, which provides an insight for optimizing over the FBL rates by the NA.
\begin{lemma}\label{lem-2}
Consider the function 
\begin{equation*}
f(\gamma)=\ln\left(1+\gamma\right)-a\sqrt{\frac{\gamma}{1+\gamma}},
\end{equation*}
where $\gamma\geq 0$ is a variable, and $a>0$ is a given and constant parameter. 
Then, $f(\gamma)$ is minimized at $\gamma^\star$, where $f(\gamma^\star)<0$. 
Moreover, $f(\gamma)$ is strictly decreasing in $0\leq\gamma<\gamma^\star$, and strictly increasing for $\gamma>\gamma^\star$ (see Fig. \ref{Fig-3}). Additionally, $f(\gamma)$ has two root given by $\gamma=0$, and $\gamma=\gamma^{(0)}$.
\end{lemma}
\begin{proof}
The derivative of $f(\gamma)$ with respect to $\gamma$ is 
\begin{equation*}
\frac{\partial f}{\partial \gamma}=\frac{1}{1+\gamma}-\frac{a}{2}\left(\frac{\gamma}{1+\gamma}\right)^{-\frac{1}{2}}
\left(\frac{1}{1+\gamma}\right)^{{2}},
\end{equation*}
which can be simplified as
\begin{equation*}
\frac{\partial f}{\partial \gamma}=\frac{1}{1+\gamma}\left[1-\frac{a}{2}
\frac{1}{\sqrt{\gamma(1+\gamma)}}
\right]. 
\end{equation*}
The function $\frac{1}{\sqrt{\gamma(1+\gamma)}}$ is strictly decreasing in $\gamma$ and takes values $\infty$ for $\gamma=0$ and $0$ for $\gamma=\infty$, which proves the lemma.
\end{proof}
 The achievable rate with FBL in \eqref{rate} is a difference of two terms, which are functions of $\gamma_{lk}$. It should be emphasized that the expression in \eqref{rate} is indeed an approximation for the actual achievable rate and may not be accurate in very low SINR regimes and/or for a very short packet length and/or low decoding error probability \cite{erseghe2016coding, erseghe2015evaluation, cocskun2019efficient, lancho2019single, lancho2020saddlepoint}. 
According to Lemma \ref{lem-2}, $r_{lk}$ can be negative for very low $\gamma_{lk}$ (see Fig. \ref{Fig-3}), which implies that \eqref{rate} is not an accurate approximation for $\gamma_{lk}\ll 1$, which is also confirmed by the results in \cite{lancho2020saddlepoint}. 
We refer the reader to \cite[Sec. IV.C]{polyanskiy2010channel} for detailed discussions on the accuracy of the NA with different parameters.

The EE of a user is defined as the ratio between its data rate and the power consumption for transmitting the data as \cite{zappone2015energy}
\begin{equation}
e_{lk}=\frac{r_{lk}}{p_c+\eta\mathbf{x}_{lk}^H\mathbf{x}_{lk}},
\end{equation}
where $\eta^{-1}$ is the power efficiency of each BS, and $p_c$ is the constant power consumption in the network for transmitting data to a user, which is given by \cite[Eq. (27)]{soleymani2022improper}. 
Another metric for EE is the global EE (GEE), which considers the performance of the whole network. The GEE is defined as ratio between the total achievable rate and the total power consumption of the network, i.e., \cite{zappone2015energy}
\begin{equation}
GEE=\frac{\sum_{lk}r_{lk}}{LKp_c+\eta\sum_{lk}\mathbf{x}_{lk}^H\mathbf{x}_{lk}}.
\end{equation}

\subsection{Problem statement}
We consider a general optimization problem, which includes a large variety of optimization problems. This general optimization problem can be formulated as
\begin{subequations}\label{ar-opt}
\begin{align}
 \underset{\{\mathbf{x}\}\in\mathcal{X},\{\bm{\Theta}\}\in\mathcal{T}
 }{\max}  & 
  f_0\!\left(\left\{\mathbf{x}\right\}\!,\!\{\bm{\Theta}\}\right)\!\!\!\! &
  \text{s.t.}   \,&  f_i\left(\left\{\mathbf{x}\right\}\!,\!\{\bm{\Theta}\}\right)\geq0,\,\forall i,
\\
\label{4-c}
 &&&
r_{lk}\geq r^{th},\,\,\,\,\forall l,k,
 \end{align}
\end{subequations}
where $\mathcal{X}$ and $\mathcal{T}$ are, respectively, the feasibility sets for $\{\mathbf{x}\}$ and $\{\bm{\Theta}\}$. The feasibility set $\mathcal{X}$ is
\begin{equation}
\mathcal{X}=\left\{\{\mathbf{x}\}:\sum_{\forall k}\mathbf{x}_{lk}^H\mathbf{x}_{lk}\leq p_l,\forall l,k \right\},
\end{equation}
where $p_l$ is the power budget for BS $l$.
The constraint \eqref{4-c} is related to the latency constraint of the system. If the maximum tolerable latency is $t_c$ seconds, then the minimum achievable rate of each user per bandwidth unit should be greater than $r^{th}= \frac{n_t}{t_cw}$, where $w$ is the channel bandwidth in Hz, and $n_t$ is the packet length in bits.
Moreover, $f_i$s for all $i$ are linear functions of rates/EEs. 
Note that
 $f_i$  can also include some other linear/quadratic/convex/concave constraints in beamforming vectors/channels such as an energy harvesting constraint and/or the so-called interference temperature. However, due to a space restriction, we do not consider such constraints in this work and leave them for a future study. 

In this paper, we aim at proposing a unified optimization framework to solve the general optimization problem \eqref{ar-opt}, which  includes, for example, maximization of weighted-sum rate, maximization of minimum-weighted rates, global EE, and maximization of the minimum-weighted EE. Note that it could be possible to focus on a specific optimization problem and provide a detailed solution for it with different approaches, each with a different behavior in terms of complexity and optimality. 
However, it is also very important to concentrate on a general methodology for solving various problems with different utility/cost functions. In this work, we prefer to choose the second approach especially since the performance of RIS (either regular or STAR) should be further investigated in FBL regimes, and a general optimization framework can provide an effective tool to do so.

\section{Generalized optimization framework for systems without RIS}\label{sec-iii}
In this section, we consider the optimization of the beamforming vectors $\{\mathbf{x}\}$ for systems without RIS. Note that the algorithms in this section can be applied to RIS-assisted systems when the reflecting coefficients are fixed. 
The considered optimization problem in this section is 
 \begin{subequations}\label{ar-opt-x}
\begin{align}
\underset{\{\mathbf{x}\}\in\mathcal{X}
 }{\max}  & 
  f_0\left(\left\{\mathbf{x}\right\}\right) &
  \text{s.t.}   \,\,&  f_i\left(\left\{\mathbf{x}\right\}\right)\geq0,\,\forall i,
  \\
\label{26b}
 &&&
r_{lk}\geq r^{th},\,\,\,\,\forall l,k,
 \end{align}
\end{subequations}
Unfortunately, \eqref{ar-opt-x} is not convex since the rates are not concave in $\{\mathbf{x}\}$.
To solve \eqref{ar-opt-x}, we employ majorization minimization (MM), which is an iterative algorithm that consists of two steps in each iteration: majorization and minimization \cite{sun2017majorization}. 
In the majorization step, the non-convex constraints (and/or objective function) are approximated by suitable surrogate functions. 
Note that the surrogate functions in MM algorithms should fulfill three specific conditions mentioned in, e.g., \cite[Sec. III]{soleymani2020improper}. 
In the minimization step, the corresponding surrogate optimization problem is solved. 
Note that MM has many applications not only in communications, but also in machine learning, signal processing, and other research areas \cite{sun2017majorization}.  In this paper, we specialize MM to solve \eqref{ar-opt-x} in the context of STAR-RIS-assisted URLLC systems. 
In general, 
MM starts with a feasible initial point and converges to a stationary point of the considered optimization problem, which meets the first order optimality constraint and satisfies the Karush-Kuhn-Tucker (KKT) conditions \cite{lanckriet2009convergence}. Since we employ MM, our proposed framework for systems without RIS converges to a stationary point of \eqref{ar-opt-x}. The final point of MM-based algorithms may depend on the initial point, which should be feasible and can be chosen either randomly or heuristically \cite{lipp2016variations}. 
As indicated in Section \ref{sec=re}, the rates in \eqref{rate} may not be accurate for very low SINR, i.e, $\gamma_{lk}\ll 1$. 
Additionally, the latency constraint in \eqref{26b} may not be satisfied for a random initial point.
Thus, to obtain a feasible and suitable initial point, 
we can employ the approach in Appendix \ref{simp-sec}, which takes the solution of the maximization of the minimum SINR of users as an initial point. 

In the proposed framework, we firstly approximate the rates by suitable concave lower-bound surrogate functions and then, solve the corresponding surrogate optimization problem. 
To this end, we employ the lower-bounds in the following lemma.
\begin{lemma}\label{lem-1}
A concave lower bound for $r_{lk}$ is
\begin{multline}\label{eq=lem=q}
r_{lk}\geq \tilde{r}_{lk}=
a_{lk}
+
\frac{2\mathfrak{R}
\left\{
\left(
\mathbf{h}_{lk,l}
\mathbf{x}_{lk}^{(t-1)}
\right)^*
\mathbf{h}_{lk,l}
\mathbf{x}_{lk}
\right\}
}{
\sigma^2
+
\sum_{[ij]\neq [lk]}\left|\mathbf{h}_{lk,i}
\mathbf{x}_{ij}^{(t-1)}\right|^2
}
\\
+
\frac{2Q^{-1}(\epsilon^c)}{\sqrt{n_tV_{lk}^{(t-1)}}}
\frac{
\sigma^2\!+\!\sum_{[ij]\neq [lk]}\mathfrak{R}\!\left\{\!\left(\mathbf{h}_{lk,i}
\mathbf{x}_{ij}^{(t-1)}\right)^*\!
\mathbf{h}_{lk,i}
\mathbf{x}_{ij}\!
\right\}
}{\sigma^2+\sum_{ij}\left|\mathbf{h}_{lk,i}
\mathbf{x}_{ij}^{(t-1)}\right|^2}
\\
-
b_{lk}
\frac{
\sigma^2
+
\sum_{ij}\left|\mathbf{h}_{lk,i}
\mathbf{x}_{ij}\right|^2
}{
\sigma^2
+
\sum_{ij}\left|\mathbf{h}_{lk,i}
\mathbf{x}_{ij}^{(t-1)}\right|^2
},
\end{multline}
where $t$ is the number of the current iteration, $a_{lk}$, and $b_{lk}$ 
are constants and, respectively, give by
\begin{align*}
a_{lk}&=\log\left(1+\gamma_{lk}^{(t-1)}\right)
-
\gamma_{lk}^{(t-1)}
\\
&\hspace{1cm}-
\frac{Q^{-1}(\epsilon^c)}{\sqrt{n_t}}\left(\frac{\sqrt{V_{lk}^{(t-1)}}}{2}
+\frac{1}{\sqrt{V_{lk}^{(t-1)}}}\right),
\\
b_{lk}&=\gamma_{lk}^{(t)}+\frac{\zeta_{lk}^{(t-1)}Q^{-1}(\epsilon^c)}{\sqrt{n_tV_{lk}^{(t-1)}}},
\end{align*}
where $\gamma_{lk}^{(t-1)}$ and $V_{lk}^{(t-1)}$ are, respectively, obtained by replacing $\{\mathbf{x}^{(t-1)}\}$ in \eqref{sinr} and \eqref{dis-}. Moreover, 
\begin{equation*}
\zeta_{lk}^{(t-1)}
=
\frac{
\sigma^2+\sum_{[ij]\neq [lk]}|\mathbf{h}_{lk,i}
\mathbf{x}_{ij}^{(t-1)}|^2
}{\sigma^2+\sum_{ij}|\mathbf{h}_{lk,i}
\mathbf{x}_{ij}^{(t-1)}|^2}.
\end{equation*}
The concave lower bound in \eqref{eq=lem=q} is quadratic in $\{\mathbf{x}\}$, and is obtained by deriving a concave lower bound for the Shannon rate as well as by finding a convex upper bound for $\delta_{lk}$.
\end{lemma}
\begin{proof}
Please refer to Appendix \ref{app=b}.
\end{proof}
Note that the concave lower-bound rates $\tilde{r}_{lk}$ for all $l,k$ are quadratic in $\mathbf{x}_{ij}$ for all $i,j$. 
Substituting $\tilde{r}_{lk}$s in $f_i$s gives the surrogate functions $\tilde{f}_i$s and consequently, the following surrogate optimization problem
\begin{subequations}\label{ar-opt-x-sur}
\begin{align}
 \underset{\{\mathbf{x}\}\in\mathcal{X}
 }{\max}\,\,  & 
  \tilde{f}_0\left(\left\{\mathbf{x}\right\}\right) &
 \,\, \text{s.t.}   \,\,& \tilde{f}_i\left(\left\{\mathbf{x}\right\}\right)\geq0,&\forall i,
\\
 &&&
\tilde{r}_{lk}\geq r_{lk}^{th},
&\forall l,k.
 \end{align}
\end{subequations}
This optimization problem is convex for spectral efficiency metrics, which can be efficiently solved by numerical tools.  
However, \eqref{ar-opt-x-sur} is not convex if we consider EE metrics. In this case, we can obtain the global optimal solution of \eqref{ar-opt-x-sur} by Dinkelbach-based algorithms \cite{zappone2015energy}. Note that the proposed framework converges to a stationary point of \eqref{ar-opt} since it falls into MM.
In the following subsections, we will discuss the solutions of \eqref{ar-opt-x-sur} with different utility functions. 
To this end, we select the minimum-weighted rate, weighted-sum rate, global EE, and minimum-weighted EE as utility functions since they are among the most important metrics in practice as well as in the literature either with Shannon rate or in FBL regimes \cite{ghanem2020resource, nasir2020resource, soleymani2020rate, nasir2021cell, xu2022rate, Sole1909:Energy, shehab2021effective, soleymani2021distributed, soleymani2022rate, ou2022energy}.
\subsection{Minimum-weighted rate maximization}
The maximization of the  minimum-weighted rate can be written as 
\begin{align}
\max_{\{\mathbf{x}\}\in\mathcal{X},r}&r&
\text{s.t.}\,\,&r_{lk}\geq \max\left(\lambda_{lk}r,r^{th}\right)&\forall l,k,
\end{align}
where $\lambda_{lk}$s are the weights corresponding to the priorities assigned to the users.  Employing the lower bounds in Lemma \ref{lem-1}, we have the following surrogate optimization problem 
\begin{align}\label{eq30}
\max_{\{\mathbf{x}\}\in\mathcal{X},r}&r&
\text{s.t.}\,\,&\tilde{r}_{lk}\geq \max\left(\lambda_{lk}r,r^{th}\right)&\forall l,k,
\end{align}
which is convex and can be efficiently solved. 
\subsection{Weighted-sum rate maximization}
The weighted-sum-rate maximization problem is 
\begin{align}\label{sum-opt}
\max_{\{\mathbf{x}\}\in\mathcal{X}}&\sum_{lk}\lambda_{lk}r_{lk}&
\text{s.t.}\,\,&r_{lk}\geq r^{th}_{lk}&\forall l,k,
\end{align}
where $r^{th}_{lk}$ is the minimum required rate for u$_{lk}$. 
The problem \eqref{sum-opt} is not convex, but can be solve by the proposed optimization framework. That is, we replace the rates by the surrogate functions in Lemma \ref{lem-1}, which yields the following convex problem
\begin{align}\label{sum-opt-2}
\max_{\{\mathbf{x}\}\in\mathcal{X}}&\sum_{lk}\lambda_{lk}\tilde{r}_{lk}&
\text{s.t.}\,\,&\tilde{r}_{lk}\geq r^{th}_{lk}&\forall l,k.
\end{align}

\subsection{Global energy efficiency maximization}
The GEE maximization problem can be written as
\begin{align}\label{gee-opt}
\max_{\{\mathbf{x}\}\in\mathcal{X}}&\frac{\sum_{lk}r_{lk}}{LKp_c+\eta\sum_{lk}\mathbf{x}_{lk}^H\mathbf{x}_{lk}}&
\text{s.t.}\,\,&r_{lk}\geq r^{th}_{lk}&\forall l,k.
\end{align}
Replacing the rates by the lower bounds in Lemma \ref{lem-1}, we have the following surrogate function
\begin{align}\label{gee-opt-sur}
\max_{\{\mathbf{x}\}\in\mathcal{X}}&\frac{\sum_{lk}\tilde{r}_{lk}}{LKp_c+\eta\sum_{lk}\mathbf{x}_{lk}^H\mathbf{x}_{lk}}&
\text{s.t.}\,\,&\tilde{r}_{lk}\geq r^{th}_{lk}&\forall l,k,
\end{align}
which is non-convex and falls into fractional-programming problems. 
We can obtain the global optimal solution of \eqref{gee-opt-sur} by the Dinkelbach algorithm since the numerator of the objective function is concave in $\mathbf{x}_{lk}$ while its denominator is convex in $\mathbf{x}_{lk}$ \cite{zappone2015energy}.
The global optimal solution of \eqref{gee-opt-sur} can be obtained by iteratively solving \cite{zappone2015energy}
\begin{subequations}\label{gee-opt-sur-2}
\begin{align}
\max_{\{\mathbf{x}\}\in\mathcal{X}}&
{\sum_{lk}\tilde{r}_{lk}}
-\mu^{(t,q)}
\left({LKp_c+\eta\sum_{lk}\mathbf{x}_{lk}^H\mathbf{x}_{lk}}
\right)\\
\text{s.t.}\,\,&\tilde{r}_{lk}\geq r^{th}_{lk}\hspace{1.5cm}\forall l,k,
\end{align}
\end{subequations}
and updating $\mu^{(t,q)}$ as
\begin{equation}
\mu^{(t,q)}=\frac{\sum_{lk}\tilde{r}_{lk}\left(\mathbf{x}_{lk}^{(t,q)}\right)}{LKp_c+\eta\sum_{lk}\mathbf{x}_{lk}^{(t,q)^H}\mathbf{x}_{lk}^{(t,q)}},
\end{equation}
where $\mathbf{x}_{lk}^{(t,q)}$ is the initial point at the $q$-th iteration of the Dinkelbach algorithm, which is the solution of the previous step.  
We refer the readers to \cite{zappone2015energy} for a detailed survey on fractional programming and the Dinkelbach algorithm. 
\subsection{ Minimum-weighted energy efficiency maximization}
The minimum-weighted EE maximization can be written as
\begin{subequations}
\begin{align}
\max_{\{\mathbf{x}\}\in\mathcal{X},e}&e&
\text{s.t.}\,\,&e_{lk}\geq e&\forall l,k,
\\
&&&{r}_{lk}\geq r^{th}_{lk}&\forall l,k.
\end{align}
\end{subequations}
Employing the optimization framework, we have the following surrogate optimization problem
\begin{subequations}\label{sur-ee-opt}
\begin{align}
\max_{\{\mathbf{x}\}\in\mathcal{X},e}&e&
\text{s.t.}\,\,&\tilde{e}_{lk}=\frac{\tilde{r}_{lk}}{p_c+\eta\mathbf{x}_{lk}^H\mathbf{x}_{lk}}\geq e&\forall l,k,
\\
&&&\tilde{r}_{lk}\geq r^{th}_{lk}&\forall l,k,
\end{align}
\end{subequations}
which is non-convex, but its global optimal solution can be obtained by the generalized Dinkelbach algorithm (GDA) since  $\tilde{e}_{lk}$ has a concave numerator and convex denominator. 
Applying GDA, the global optimal solution of \eqref{sur-ee-opt} can be obtained by iteratively solving 
\begin{subequations}\label{ee-opt-sur-2}
\begin{align}
\max_{\{\mathbf{x}\}\in\mathcal{X}}&
e\\
\text{s.t.}\,\,&
\tilde{r}_{lk}-
\mu^{(t,q)}
\left(
p_c+\eta\mathbf{x}_{lk}^H\mathbf{x}_{lk}
\right)
\geq \alpha_{lk}e&\forall l,k,
\\
&\tilde{r}_{lk}\geq r^{th}_{lk}&\forall l,k,
\end{align}
\end{subequations}
and updating $\mu^{(t,q)}$ as
$\mu^{(t,q)}=\min_{lk}\left\{\tilde{e}_{lk}\left(\mathbf{x}_{lk}^{(t,q)}\right)\right\}$, 
where $\mathbf{x}_{lk}^{(t,q)}$ is the initial point at the $q$-th iteration of the Dinkelbach algorithm, which is the solution of the previous step.  

\section{Extending the optimization framework to (STAR-)RIS-assisted systems}\label{seciv}
In this section, we extend the framework in Section \ref{sec-iii} to solve \eqref{ar-opt} by MM and alternating optimization (AO) for RIS-assisted systems. 
To this end, at the $t$-th iteration, we first fix the reflecting coefficients to $\{\bm{\Theta}^{(t-1)}\}$ and optimize over the beamforming vectors to obtain $\{\mathbf{x}^{(t)}\}$. 
We then fix  the beamforming vectors $\{\mathbf{x}^{(t)}\}$ and update the reflecting coefficients as $\{\bm{\Theta}^{(t)}\}$. We iterate the procedure until the solution converges. Note that our proposed framework converges since it produces a non-decreasing sequence in the objective function $f_0$. If the feasibility set $\mathcal{T}$ is convex, the framework falls into MM and converges to a stationary point of \eqref{ar-opt}. 
For the initial point of the schemes, we can 
employ the scheme in Appendix \ref{simp-sec}, similar to Section \ref{sec-iii}.
 In the following subsections, we provide the solutions for updating $\{\mathbf{x}\}$ and $\{\bm{\Theta}\}$.

\subsection{Optimizing the beamforming vectors}
In this subsection, we assume that the reflecting coefficients are fixed to $\{\bm{\Theta}^{(t-1)}\}$, and we optimize over $\{\mathbf{x}\}$, which results in 
\begin{subequations}\label{ar-opt-x-2}
\begin{align}%
 \underset{\{\mathbf{x}\}\in\mathcal{X}
 }{\max}  & 
  f_0\!\left(\!\left\{\mathbf{x}\right\}\!,\!\{\bm{\Theta}^{(t-1)}\}\!\right)\!\!\!\!\! &
  \text{s.t.}   \,&  f_i\!\left(\!\!\left\{\mathbf{x}\right\}\!,\!\{\bm{\Theta}^{(t-1)}\}\!\right)\!\!\geq0,\forall i,
  \\
 &&&
{r}_{lk}\geq r_{lk}^{th},
\,\,\forall l,k.
 \end{align}
\end{subequations}
This problem can be solved similar to the framework in  Section \ref{sec-iii}. 
Since it is straightforward to modify the framework in  Section \ref{sec-iii} to solve \eqref{ar-opt-x-2}, we do not repeat the solution here.
\subsection{Optimizing the reflecting coefficients}\label{sec-opt-ref}
In this subsection, we assume that the beamforming vectors are fixed to $\{\mathbf{x}^{(t)}\}$, and we optimize only over the reflecting coefficients $\{\bm{\Theta}\}$. 
In other words, we want to solve 
\begin{subequations}\label{ar-opt-t}
\begin{align}
 \underset{\{\bm{\Theta}\}\in\mathcal{T}
 }{\max}  & 
  f_0\!\left(\!\left\{\mathbf{x}^{(t)}\!\right\}\!,\{\bm{\Theta}\}\right)\!\!\!\!\!  &
  \text{s.t.}   \,&  f_i\left(\left\{\mathbf{x}^{(t)}\right\}\!,\{\bm{\Theta}\}\right)\!\geq\!0,\,\forall i,
  \\
 &&&
{r}_{lk}\geq r_{lk}^{th},
\,\,\forall l,k,
 \end{align}
\end{subequations}
 where $f_i$s are linear functions of the rates for both spectral and energy efficiency metrics since the beamforming vectors are fixed.
To solve \eqref{ar-opt-t}, we first obtain suitable concave lower bounds for the rates. We then mention how to transform $\mathcal{T}$ into a convex
problem if $\mathcal{T}$ is not a convex set. 
The rates have the same structure in the channels as in the beamforming vectors $\{\mathbf{x}\}$. Thus, we can employ a similar concave lower bound for the rates. To avoid notational confusions, we restate the lower bounds for the rates in the following corollary. 
\begin{corollary}\label{cor-1}
A concave lower bound for $r_{lk}$ is
\begin{multline}\label{eq42}
r_{lk}\geq \hat{r}_{lk}=
a_{lk}
+
\frac{2\mathfrak{R}
\left\{
\left(
\mathbf{h}_{lk,l}^{(t-1)}
\mathbf{x}_{lk}^{(t)}
\right)^*
\mathbf{h}_{lk,l}
\mathbf{x}_{lk}^{(t)}
\right\}
}{
\sigma^2
+
\sum_{[ij]\neq [lk]}\left|\mathbf{h}_{lk,i}^{(t-1)}
\mathbf{x}_{ij}^{(t)}\right|^2
}
\\
+\!
\frac{2Q^{-1}(\epsilon^c)}{\sqrt{n_tV_{lk}^{(t-1)}}}
\frac{
\sigma^2\!+\!\sum_{[ij]\neq [lk]}\mathfrak{R}\!\left\{\!\left(\mathbf{h}_{lk,i}^{(t-1)}
\mathbf{x}_{ij}^{(t)}\!\right)^*
\!\mathbf{h}_{lk,i}
\mathbf{x}_{ij}^{(t)}
\right\}
}{\sigma^2+\sum_{ij}\left|\mathbf{h}_{lk,i}^{(t-1)}
\mathbf{x}_{ij}^{(t)}\right|^2}
\\
-
b_{lk}
\frac{
\sigma^2
+
\sum_{ij}\left|\mathbf{h}_{lk,i}
\mathbf{x}_{ij}^{(t)}\right|^2
}{
\sigma^2
+
\sum_{ij}
\left|\mathbf{h}_{lk,i}^{(t-1)}
\mathbf{x}_{ij}^{(t)}\right|^2
},
\end{multline}
where the parameters are defined as in Lemma \ref{lem-1}.
\end{corollary}
Note that the concave lower bound in Corollary \ref{cor-1} is quadratic in $\{\bm{\Theta}\}$. Plugging the concave lower bound $\hat{r}_{lk}$ into \eqref{ar-opt-t}, we have the following surrogate optimization problem 
\begin{subequations}\label{ar-opt-t-sur}
\begin{align}
 \underset{\{\bm{\Theta}\}\in\mathcal{T}
 }{\max}  & 
  \hat{f}_0\left(\!\left\{\!\mathbf{x}^{(t)}\!\right\}\!,\!\{\bm{\Theta}\}\!\right)\!\! &
  \text{s.t.}   \,&  \hat{f}_i\left(\!\left\{\!\mathbf{x}^{(t)}\!\right\}\!,\{\bm{\Theta}\}\!\right)\!\geq0,\,\forall i,
  \\
 &&&
\hat{r}_{lk}\geq r_{lk}^{th},
\,\,\forall l,k.
 \end{align}
\end{subequations}
The optimization problem \eqref{ar-opt-t-sur} is convex if the feasibility set $\mathcal{T}$ is convex, i.e., when considering $\mathcal{T}_U$ for regular RIS and $\mathcal{T}_{SU}$ for STAR-RIS. In this case, the proposed framework converges to a stationary point of \eqref{ar-opt}.  
Unfortunately, the feasibility sets $\mathcal{T}_I$, $\mathcal{T}_C$, $\mathcal{T}_{SI}$, and $\mathcal{T}_{SN}$ are  not convex. 
To convexify them, we employ a suboptimal approach 
as described in the following. 
It should be emphasized that although the convergence of our proposed framework is guaranteed for all the RIS feasibility sets, we do not make any claim on the optimality of our framework for the non-convex RIS feasibility sets.
\subsubsection{Feasibility set $\mathcal{T}_I$} 
In this case, we have the non-convex constraint $|\theta_{mn}|=1$ for all $m,n$, which can be rewritten as
\begin{align}\label{eq-32}
|\theta_{mn}|^2&\leq 1,\\
|\theta_{mn}|^2&\geq 1,\label{eq-33}
\end{align}
for all $m,n$. The constraint \eqref{eq-32} is  convex. However, the constraint \eqref{eq-33} makes the problem \eqref{ar-opt-t-sur} non-convex since $|\theta_{mn}|^2$ is a convex function, rather than being concave. 
Thus, we employ convex-concave procedure (CCP) and approximate \eqref{eq-33} by a linear constraint as in \cite[Eq. (38)]{soleymani2022noma}.
We then also relax the constraint in \cite[Eq. (39)]{soleymani2022noma} by introducing $\epsilon>0$ as
\begin{equation}\label{+=+}
|\theta_{mn}|^2\!\geq\!|\theta_{mn}^{(t-1)}|^2\!-2\mathfrak{R}\{\theta_{mn}^{(t-1)^*}(\theta_{mn}-\theta_{mn}^{(t-1)})\}\!\geq 1-\epsilon,
\end{equation}
for all $m,n$. 
Plugging \eqref{eq-32} and \eqref{+=+} into \eqref{ar-opt-t-sur}, we have the following convex optimization problem
 \begin{subequations}\label{ar-opt-2-rf-2-i}
\begin{align}
 \underset{\{\bm{\Theta}\}
 }{\max}  & 
  \hat{f}_0\left(\!\left\{\!\mathbf{x}^{(t)}\!\right\}\!,\{\bm{\Theta}\}\!\right)\!\!\! &
  \text{s.t.}   \,\,&  \hat{f}_i\left(\!\left\{\mathbf{x}^{(t)}\!\right\}\!,\{\bm{\Theta}\}\!\right)\!\geq\!0,\,\forall i, 
\\
 &&&
\hat{r}_{lk}\geq r_{lk}^{th},
\,\,\forall l,k,
    \\
    &&&\eqref{+=+},\eqref{eq-32}\,\,\forall m,n.
 \end{align}
\end{subequations}
Since \eqref{ar-opt-2-rf-2-i} is convex, it can be solved efficiently by numerical tools. Let us call the solution of \eqref{ar-opt-2-rf-2-i} as $\{\hat{\bm{\Theta}}\}$. Although the relaxation in \eqref{+=+} makes the convergence of our framework faster, it may make  $\{\hat{\bm{\Theta}}\}$ infeasible.
 To obtain a feasible point, we project  $\{\hat{\bm{\Theta}}\}$ to $\mathcal{T}_I$ by normalizing $\{\hat{\bm{\Theta}}\}$ as $\{\hat{\bm{\Theta}}^{\text{new}}\}$, i.e.,
\begin{equation}\label{eq==27}
\hat{\theta}_{mn}^{\text{new}}=\frac{\hat{\theta}_{mn}}{|\hat{\theta}_{mn}|}, \hspace{1cm}\forall m,n.
\end{equation}
 To ensure the convergence of our scheme, we update $\{\bm{\Theta}\}$ as
\begin{equation}\label{eq-42}
\{\bm{\Theta}^{(t)}\}=
\left\{
\begin{array}{lcl}
\{\hat{\bm{\Theta}}^{\text{new}}\}&\text{if}&
f\left(\left\{\mathbf{P}^{(t)}\right\},\{\hat{\bm{\Theta}}^{\text{new}}\}\right)\geq
\\
&&
f\left(\left\{\mathbf{P}^{(t)}\right\},\{\bm{\Theta}^{(t-1)}\}\right)
\\
\{\bm{\Theta}^{(t-1)}\}&&\text{Otherwise}.
\end{array}
\right.
\end{equation}
This updating rule guarantees convergence by generating a sequence of non-decreasing objective functions.

\subsubsection{Feasibility set $\mathcal{T}_C$}
To convexifying  $\mathcal{T}_{C}$, we first relax the relationship between the phase and amplitude of reflecting components. In other words, we consider them as independent optimization parameters. This relaxation yields \eqref{eq-32} and
\begin{align}
 |\theta_{mn}|^2&\geq|\theta|_{\min}^2,
\label{eq=9-20}
\end{align}
for all $m,n$. 
Now, the problem is similar to convexifying the feasibility set $\mathcal{T}_C$.
That is, we employ CCP to find a suitable linear lower bound for $|\theta_{mn}|^2$ as
\begin{equation}\label{eq-50-2}
|\theta_{mn}^{(t-1)}|^2+2\mathfrak{R}\left(\theta_{mn}^{(t-1)}(\theta_{mn}-\theta_{mn}^{(t-1)})^*\right)\geq |\theta|_{\min}^2,
\end{equation}
which results in the following convex surrogate optimization problem
 \begin{subequations}\label{ar-opt-2-rf-2-c}
\begin{align}
 \underset{\{\bm{\Theta}\}
 }{\max}  & 
  \hat{f}_0\left(\!\left\{\mathbf{x}^{(t)}\!\right\}\!,\{\bm{\Theta}\}\!\right)\!\!\! &
  \text{s.t.}   \,&  \hat{f}_i\left(\!\left\{\mathbf{x}^{(t)}\!\right\}\!,\!\{\bm{\Theta}\}\!\right)\!\geq\!0,\,\forall i, 
\\
 &&&
\hat{r}_{lk}\geq r_{lk}^{th},
\,\,\forall l,k,
    \\
    &&&\eqref{eq-50-2},\eqref{eq-32}\,\,\forall m,n.
 \end{align}
\end{subequations}
Due to the relaxation of the dependency of the amplitude and phase of reflecting components, the solution of \eqref{ar-opt-2-rf-2-c}, i.e., $\{\bm{\Theta}^{{(\star)}}\}$ may be infeasible. To generate a feasible solution, we project $\{\bm{\Theta}^{{(\star)}}\}$ into $\mathcal{T}_{C}$ as 
\begin{equation}
\{\hat{\bm{\Theta}}^{\text{new}}\}=\mathcal{F}(\angle\{\bm{\Theta}^{{(\star)}}\}),
\end{equation}
where $\mathcal{F}$ is defined as in \eqref{eq*=*},
and update $\{\bm{\Theta}\}$ according to \eqref{eq-42}, which guarantees the convergence  of the algorithm. 

\subsubsection{STAR-RIS with mode switching and feasibility set $\mathcal{T}_{SI}$ or $\mathcal{T}_{SN}$}
In our proposed mode switching (MS) approach, we randomly divide the reflecting components into two sets: reflecting set and transmitting set. This scheme converts each STAR-RIS into two regular RISs. Thus, the MS scheme can be obtained similar to considering two RISs with the feasibility set $\mathcal{T}_{I}$, instead of each STAR-RIS. 

\subsubsection{STAR-RIS with time switching and feasibility set $\mathcal{T}_{SI}$ or $\mathcal{T}_{SN}$}
In our proposed time switching (TS) approach, we divide each time slot into two sub-slots. 
In the first sub-slot, all the RIS components operate in the reflection mode, while in the second sub-slot, they all operate in the transmission mode.
In both sub-slots, each STAR-RIS operates similar to a regular RIS with feasibility set $\mathcal{T}_{I}$. Thus, we can employ the proposed algorithm for regular RIS with $\mathcal{T}_{I}$ to update the RIS components in each sub-slot. 

\subsubsection{Feasibility set $\mathcal{T}_{SI}$ (STAR-RIS with energy splitting)} 
In the energy splitting (ES) approach, each RIS component can simultaneously reflect and transmit, which makes it impossible to model STAR-RIS as a set of regular RISs. Thus, we have to directly tackle 
the constraint $|\theta_{mn}^t|^2+|\theta_{mn}^r|^2=1$, which can be rewritten as
\begin{align}\label{eq33}
|\theta_{mn}^t|^2+|\theta_{mn}^r|^2&\leq 1,\\
|\theta_{mn}^t|^2+|\theta_{mn}^r|^2&\geq 1.
\label{eq34}
\end{align}
The former constraint is convex, but the latter is not  since $|\theta_{mn}^t|^2$ and $|\theta_{mn}^r|^2$ are convex functions. Thus, we can employ CCP to approximate \eqref{eq34} by a linear constraint similar to the feasibility set $\mathcal{T}_{I}$. 
To make the convergence faster, we also relax \eqref{eq34} by introducing a positive variable $\epsilon$ as
\begin{multline}\label{eq-50-6}
|\theta_{mn}^{r^{(t-1)}}|^2+2\mathfrak{R}\left(\theta_{mn}^{r^{(t-1)}}(\theta_{mn}^r-\theta_{mn}^{r^{(t-1)}})^*\right)
+
|\theta_{mn}^{t^{(t-1)}}|^2
\\
+2\mathfrak{R}\left(\theta_{mn}^{t^{(t-1)}}(\theta_{mn}^t-\theta_{mn}^{t^{(t-1)}})^*\right)\geq 1-\epsilon.
\end{multline}
Substituting \eqref{eq33} and \eqref{eq-50-6} in \eqref{ar-opt-t-sur}, we have the following convex surrogate optimization problem
\begin{subequations}\label{ar-opt-2-rf-2-c=st}
\begin{align}
 \underset{\{\bm{\Theta}\}
 }{\max}  & 
  \hat{f}_0\left(\!\left\{\!\mathbf{x}^{(t)}\right\}\!,\{\bm{\Theta}\}\!\right)\!\!\! &
  \text{s.t.}   \,&  \hat{f}_i\left(\!\left\{\!\mathbf{x}^{(t)}\right\}\!,\{\bm{\Theta}\}\!\right)\!\geq\!0,\,\forall i, 
\\
 &&&
\hat{r}_{lk}\geq r_{lk}^{th},
\,\,\forall l,k,
    \\
    &&&\eqref{eq33},\eqref{eq-50-6}\,\,\forall m,n.
 \end{align}
\end{subequations}
Due to the relaxation in \eqref{eq-50-6}, the solution of \eqref{eq-50-6}, i.e., $\bm{\Theta}_m^{r^{(\star)}}$ and $\bm{\Theta}_m^{t^{(\star)}}$, may not be feasible. 
Thus, we first project $\bm{\Theta}_m^{r^{(\star)}}$ and $\bm{\Theta}_m^{t^{(\star)}}$ into $\mathcal{T}_{SI}$ as
\begin{align}\label{eq3300}
\hat{\theta}_{mn}^t&\!=\!\frac
{{\theta}_{mn}^{t^{(\star)}}}
{\sqrt{|{\theta}_{mn}^{t^{(\star)}}|^2+|{\theta}_{mn}^{r^{(\star)}}|^2}},&\!\!
\hat{\theta}_{mn}^r&\!=\!\frac
{{\theta}_{mn}^{r^{(\star)}}}
{\sqrt{|{\theta}_{mn}^{t^{(\star)}}|^2+|{\theta}_{mn}^{r^{(\star)}}|^2}},
\end{align}
for all $m,n$. 
Finally, we update $\bm{\Theta}_m^r$ and $\bm{\Theta}_m^t$ as 
\begin{equation}\label{eq-42-star}
\{\bm{\Theta}^{r^{(t)}}\!\!,\bm{\Theta}^{t^{(t)}}\!\}\!=\!\!
\left\{\!\!\!\!
\begin{array}{lcl}
\{\hat{\bm{\Theta}}^{r},\hat{\bm{\Theta}}^{t}\}\!\!&\!\!\text{if}\!\!&
f\!\!\left(\!\!\left\{\mathbf{P}^{(t)}\right\}\!,\{\hat{\bm{\Theta}}^{r},\hat{\bm{\Theta}}^{t}\}\!\right)\!\!\geq\!\!\!\!\!
\\
&&
f\!\left(\left\{\mathbf{P}^{(t)}\right\},\{\bm{\Theta}^{(t-1)}\}\right)
\\
\{\bm{\Theta}^{(t-1)}\}&&\text{Otherwise},
\end{array}
\right.
\end{equation}
where $\{\bm{\Theta}^{(t-1)}\}=\{\bm{\Theta}^{r^{(t-1)}}\!\!,\bm{\Theta}^{t^{(t-1)}}\}$.
This updating rule ensures the convergence.

\subsubsection{Feasibility set $\mathcal{T}_{SN}$ (STAR-RIS with energy splitting)}
The feasibility set $\mathcal{T}_{SN}$ is a subset of $\mathcal{T}_{SI}$. In other words, $\mathcal{T}_{SN}$ includes the two convex constraints in \eqref{st-eq-} and \eqref{st-eq-1}, in addition to the constraint $|\theta_{mn}^t|^2+|\theta_{mn}^r|^2=1$.
Since \eqref{st-eq-} and \eqref{st-eq-1} are convex constraints, we should handle only the non-convex constraint $|\theta_{mn}^t|^2+|\theta_{mn}^r|^2=1$, which can be done similar to our algorithm for the feasibility set $\mathcal{T}_{SI}$. 
That is, we have to solve the following convex surrogate problem:
\begin{subequations}\label{ar-opt-2-rf-2-c=st}
\begin{align}
 \underset{\{\bm{\Theta}\}
 }{\max}  & 
  \hat{f}_0\left(\!\left\{\!\mathbf{x}^{(t)}\right\}\!,\{\bm{\Theta}\}\!\right)\!\!\!\! &
  \text{s.t.}   \,&  \hat{f}_i\left(\!\left\{\!\mathbf{x}^{(t)}\right\}\!,\{\bm{\Theta}\}\!\right)\!\geq\!0,\,\forall i, 
\\
 &&&
\hat{r}_{lk}\geq r_{lk}^{th},
\,\,\forall l,k,
    \\
    &&& \eqref{st-eq-},\!\eqref{st-eq-2},\eqref{eq33},\eqref{eq-50-6}\,\forall m,\!n.
 \end{align}
\end{subequations}
The solution of \eqref{ar-opt-2-rf-2-c=st} might be infeasible because of the relaxation in \eqref{eq-50-6}. Thus, we normalize the solution according to \eqref{eq3300} and update $\bm{\Theta}_m^r$ and $\bm{\Theta}_m^t$ according to the rule in  \eqref{eq-42-star} to ensure the convergence.
\begin{table}[htb]
\label{alg-wsrm}
\footnotesize
\begin{tabular}{l}
\hline 
{\textbf{Algorithm I} Our algorithm for MWRM with STAR-RIS and $\mathcal{T}_{SN}$.}\\
\hline 
\hspace{0.2cm}
\textbf{Initialization}\\
\hspace{0.2cm}
{Set $\epsilon$, 
$t=1$, 
 $\{\mathbf{x}\}=\{\mathbf{x}^{(0)}\}$, and$\{\bm{\Theta}\}=\{\bm{\Theta}^{(0)}\}$ }\\
\hline 
\hspace{0.2cm}
{\textbf{While} $\left(\underset{ lk}{\min}\,\frac{r^{(t)}_{lk}}{\lambda_{lk}}-\underset{ lk}{\min}\,\frac{r^{(t-1)}_{lk}}{\lambda_{lk}}\right)/\underset{\forall l,k}{\min}\,\frac{r^{(t-1)}_{lk}}{\lambda_{lk}}\geq\epsilon$ }\\ 
\hspace{.6cm}
{{\bf Optimizing over} $\{\mathbf{P}\}$ {\bf by fixing} $\{\bm{\Theta}^{(t-1)}\}$}\\
\hspace{1.2cm}
{Obtain $\tilde{r}_{lk}^{(t-1)}$ 
based on \eqref{eq=lem=q} in Lemma \ref{lem-1}}\\
\hspace{1.2cm}
{Compute $\{\mathbf{x}^{(t)}\}$ by solving \eqref{eq30}}\\
\hspace{.6cm}
{{\bf Optimizing over} $\{\bm{\Theta}\}$ {\bf by fixing} $\{\mathbf{P}^{(t-1)}\}$}\\
\hspace{1.2cm}
{Obtain $\hat{r}_{lk}^{(t-1)}$ 
based on \eqref{eq42} in Corollary \ref{cor-1}}\\
\hspace{1.2cm}
{Compute $\bm{\Theta}_m^{r^{(\star)}}$ and $\bm{\Theta}_m^{t^{(\star)}}$ by solving \eqref{ar-opt-2-rf-2-c=st}}\\
\hspace{1.2cm} Update $\{\bm{\Theta}^{(t)}\}$ based on the rule in \eqref{eq-42-star}\\
\hspace{.6cm}
{$t=t+1$}\\
\hspace{0.2cm}
{\textbf{End (While)}}\\
\hspace{0.2cm}
{{\bf Return} $\{\mathbf{P}^{\star}\}$ and $\{\bm{\Theta}^{\star}\}$.}\\
\hline 
\end{tabular}
\normalsize
\end{table}

\subsection{Discussions on different STAR-RIS modes}
It can be expected that the  ES approach outperforms the MS and/or TS approaches since it includes the MS and/or TS approaches as special cases. In the ES approach, each RIS component can simultaneously reflect and transmit signals, while in the MS and/or TS approaches, each RIS component either transmits or reflects at a time.  
In other words, the solutions of the MS and/or TS schemes are feasible, but possibly suboptimal for the ES scheme. Indeed, in the ES scheme, it may happen that a set of RIS components work only in the transmission mode while the remaining components operate in the reflection mode, which is the same as in the MS approach.  Additionally, if we allow time slot sharing as it is the case in the TS approach, it may happen that in the ES approach, all RIS components operate in the transmission mode in the first time slot, while they all operate only in the reflection mode in the next time slot, which is the same as in the TS approach.     
As a result, an optimal ES approach never performs worse than any MS/TS scheme. 

We can also compare ES, MS and TS based on the coverage of STAR-RIS. The ES and MS can simultaneously provide a $360^\circ$ coverage. However, the TS approach can cover only a subspace (either reflection or transmission spaces), which may restrict the practicality of the TS mode especially in URLLC systems in which each user constantly requires an ultra-reliable communication with a very low latency.  Indeed, since the STAR-RIS with the TS mode can cover only a subset of users in each sub-slot, some users may not receive any signal from the STAR-RIS at a time slot. This issue is more critical when the users are close to the cell edge, which means that they may have a very weak link, and they would be in outage without the assistance of the STAR-RIS. In such scenarios and in the presence of a very stringent latency constraint, the TS approach may not be a feasible option. However, if the latency constraint is more relaxed and/or the users are not in outage without the assistance of STAR-RIS, the TS approach can be still beneficial.  Another drawback of the TS mode is that we have to solve the corresponding optimization problems twice for the coherence time of the channels, which results in  inefficient TS mode in fast fading systems. 

\subsection{Discussions on computational complexities}
In this subsection, we provide a discussion on the computational complexity of our proposed algorithms.  
Note that our schemes are iterative, and their actual computational complexities may highly depend on the implementation of the algorithms.
Here, we provide an approximated upper bound for the number of multiplications to obtain a solution for our proposed schemes.

The proposed schemes are iterative, and each iteration consist of two steps. In the first step, we obtain beamforming vectors $\{{\bf x}^{(t)}\}$, and in the second step, we compute the (STAR-)RIS coefficients $\{\bm{\Theta}^{(t)}\}$. In the following, we compute an upper bound for the number of multiplications to obtain a solution for the MWRM problem with the feasibility set $\mathcal{T}_{SU}$, which considers STAR-RIS. Since it would be very straightforward to extend the analysis to other optimization problems, we do not provide such analysis for all the considered optimization problems. 
To update the beamforming vectors in one iteration of the MWRM problem, we have to solve the convex optimization problem in \eqref{eq30}. To solve it, the total number of the Newton steps grows with the square root of the number of inequality constraints in the problem \cite[Chapter 11]{boyd2004convex}, which is equal to $LK$ for the optimization problem in \eqref{eq30}. Moreover, for each Newton step, we have to compute the surrogate functions for the rates. The number of multiplications to compute each surrogate rate function grows with $LKN_{BS}$. Therefore, the computational complexity for  solving \eqref{eq30} can be approximated as $\mathcal{O}\left(N_{BS}L^2K^2\sqrt{KL}\right)$. 
Additionally, to update the (STAR-)RIS coefficients for the MWRM, we have to solve the convex optimization problem in \eqref{ar-opt-t-sur}.
Similarly, it can be shown that the computational complexity for  solving \eqref{ar-opt-t-sur} can be approximated as $\mathcal{O}\left(N_{BS}L^2K^2\sqrt{KL+MN_{RIS}}\right)$.
Finally, setting the maximum number of the iterations to $N$, the computational complexity of our proposed scheme for the MWRM problem with the feasibility set $\mathcal{T}_{SU}$ can be approximated as 
$\mathcal{O}\left(N_{BS}L^2K^2\left(\sqrt{KL+MN_{RIS}}+\sqrt{KL}\right)\right)$.

\section{Numerical results}\label{num-sec}
\begin{figure}[t!]
    \centering
\includegraphics[width=.45\textwidth]{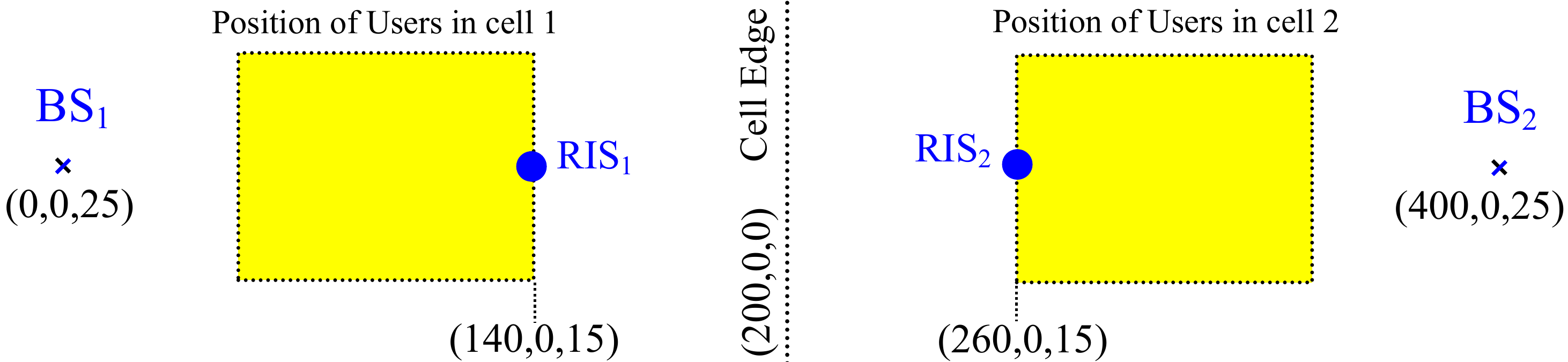}
     \caption{System topology in simulations.}
	\label{Fig-sim-model}
\end{figure}
In this section, we present some numerical results for the considered optimization problems. 
We consider a two-cell BC with $K$ users in each cell as shown in Fig. \ref{Fig-sim-model}, unless it is mentioned otherwise. 
We consider one RIS in each cell since it has been shown that a distributed implementation of RIS can outperform a collocated implementation \cite{soleymani2022improper}. 
The heights of BSs, RISs and users are, respectively, $25$, $15$, and $1.5$ meters.
The BSs are located at $(0,0,25)$ and $(400,0,25)$ while RISs  are located close to the users at $(140,0,15)$ and $(260,0,15)$. 
The $K$ users,  in each cell, are located in a square with a side $20$ meters exactly in front of the RIS.  
We represent the power budget of BSs with $P$. 
We assume that the channels through RISs are line of sight (LoS), while the direct channels are non-LoS (NLoS). It means that the links through RISs follow the Rician fading, but the direct links experience a Rayleigh fading. The propagation parameters are chosen as in \cite{soleymani2022improper}.

As indicated, to the best of our knowledge, there is no other work that considers EE metrics and/or STAR-RIS in RIS-assisted URLLC systems with FBL. 
 Thus, we consider the following schemes in the simulations:
\begin{itemize}
\item 
{\bf S-RIS}  refers to the Shannon rate in RIS-assisted systems with the feasibility set $\mathcal{T}_I$.

\item 
{\bf RIS} refers to the proposed scheme for RIS-assisted systems with the feasibility set $\mathcal{T}_I$.

\item 
{\bf RIS$_U$} (or {\bf RIS$_C$}) refers to the proposed scheme for RIS-assisted systems with the feasibility set $\mathcal{T}_U$ (or $\mathcal{T}_C$).

\item 
{\bf N} refers to the scheme without RIS.

\item 
{\bf R-RIS} refers to the proposed scheme for RIS-assisted systems with random reflecting coefficients.

\item
 {\bf ST-RIS$_{EX}$} refers to the proposed scheme for STAR-RIS-assisted systems with energy splitting and the feasibility set $\mathcal{T}_{SX}$, where $X$ can be $U$, $I$ and $N$ for modeling the feasibility set $\mathcal{T}_{SU}$, $\mathcal{T}_{SI}$, and $\mathcal{T}_{SN}$, respectively.

\item
 {\bf ST-RIS$_{M}$} (or {\bf ST-RIS$_{T}$}) refers to the proposed scheme for STAR-RIS-assisted systems with mode (or time) switching and the feasibility set $\mathcal{T}_{SI}$.
\end{itemize}

In the following, we consider the maximization of the minimum rate, sum rate, global EE and minimum EE of users  in separate subsections. 
Through numerical examples, we investigate the impact of various parameters on the performance of RIS and/or STAR-RIS. These parameters are the power budget of BSs, number of BS antennas, number of users per cell, packet length, and decoding error probability. 
We discuss how these parameters may affect on the FBL rate as well as on the gap between the FBL rate and the Shannon rate. 
\subsection{Minimum-weighted rate maximization}\label{sec5a}
In this subsection, we provide some numerical results for maximizing the minimum rate by considering the effect of different parameters. We call the minimum rate of users as the fairness rate since all users mostly achieve the same rate if we maximize the minimum rate.

\subsubsection{Impact of power budget} 

\begin{figure}
    \centering
\begin{subfigure}{0.45\textwidth}
        \centering
\includegraphics[width=\textwidth]{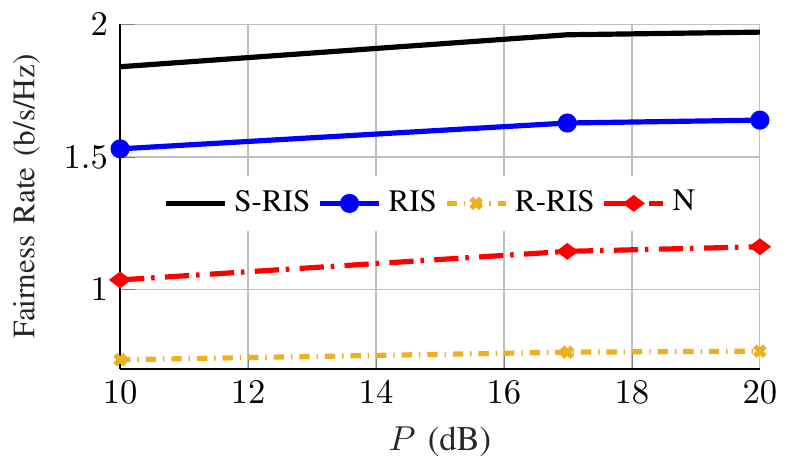}
        \caption{$N_{BS}=4$.}
    \end{subfigure}%
    \\
    \begin{subfigure}{0.45\textwidth}
        \centering
\includegraphics[width=\textwidth]{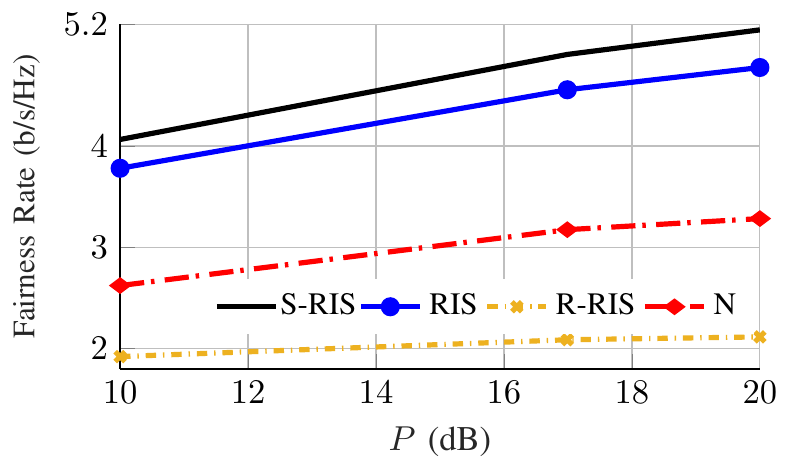}
        \caption{$N_{BS}=8$.}
    \end{subfigure}
    \caption{The average fairness rate versus $P$ for $N_{RIS}=20$, $L=2$, $K=4$, $M=2$, $n_t=256$ bits, and $\epsilon^c=10^{-5}$.}
	\label{Fig-rr}  
\end{figure}
Fig. \ref{Fig-rr} shows the average fairness rate versus the power budget of BSs for $N_{RIS}=20$, $L=2$, $K=4$, $M=2$, $n_t=256$, $\epsilon^c=10^{-5}$, and different number of BS antennas.
As can be observed, RIS can significantly increase the average fairness rate for the considered $N_{BS}$s if the reflecting coefficients are optimized properly. Interestingly, we can observe that RIS performs worse than the systems without RIS when the reflecting coefficients are chosen randomly.   
 This indicates the importance of optimizing RIS components. 
 Additionally, the performance gap between the achievable rate with FBL and the Shannon rate decreases with $N_{BS}$ when the other parameters are fixed. It happens since the SINR may be improved by increasing the number of transmit antennas. It means that the performance gap is expected to be higher when we operate in [highly] overloaded systems, i.e., when the number of users per cell is equal to or higher than the number of BS antennas.

\subsubsection{Impact of transmit antennas}
\begin{figure}[t!]
    \centering
\includegraphics[width=.45\textwidth]{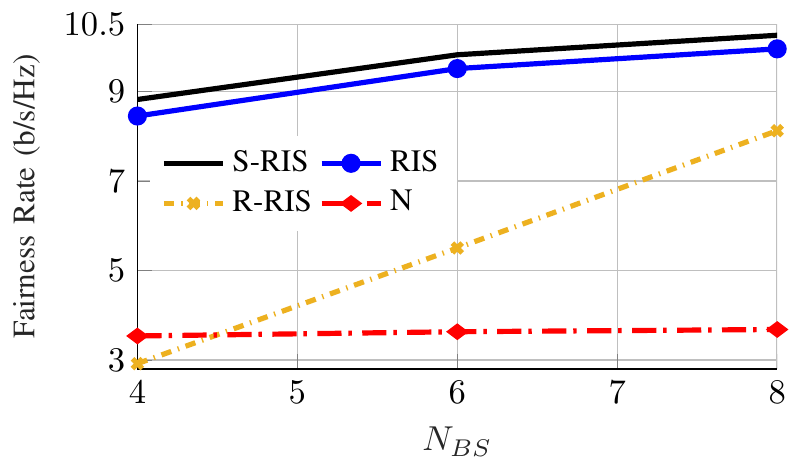}
    \caption{The average fairness rate versus $N_{BS}$ for $P=20$dB, $N_{RIS}=20$, $L=2$, $K=2$, $M=2$, $n_t=200$, and $\epsilon^c=0.001$.}
	\label{Fig-rr6} 
\end{figure}
Fig. \ref{Fig-rr6} shows the average fairness rate versus $N_{BS}$ for $P=20$dB, $N_{RIS}=20$, $L=2$, $K=2$, $M=2$, $n_t=200$, and $\epsilon^c=0.001$. As expected, the average fairness rate increases with the number of BS antennas. 
We can also observe that RIS can considerably increase the average minimum rate of users, and the benefits of RIS slightly increase with $N_{BS}$. 
Interestingly, we observe that the benefits of optimizing the reflecting coefficients decrease with $N_{BS}$ even though they are still significant. The reason is that, the effective channel gains may be more dependent on the reflecting coefficient when $N_{BS}$ is low. 
Furthermore, we observe that the gap between the Shannon rate and the FBL rate slightly decrease with $N_{BS}$, which is in line with the results in Fig. \ref{Fig-rr}. 
Moreover, we observe that the performance gap between the Shannon rate and the FBL rate is much lower than when $\epsilon=10^{-5}$ (see Fig. \ref{Fig-rr}), which implies that we should expect higher performance loss comparing to the Shannon if we require ultra-reliable communication.

\subsubsection{Impact of number of users per cell}
\begin{figure}[t!]
    \centering
\includegraphics[width=.4\textwidth]{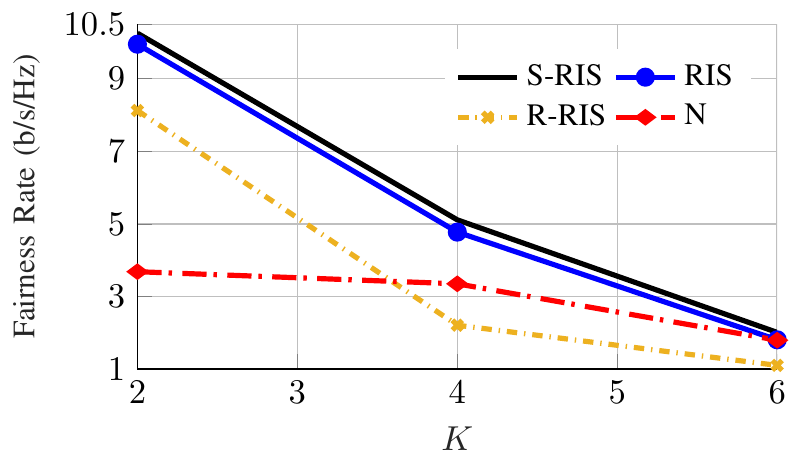}
    \caption{The average fairness rate versus $K$ for $P=10$dB, $N_{RIS}=20$, $N_{BS}=8$, $L=2$, $K=2$, $M=2$, $n_t=200$, and $\epsilon^c=0.001$.}
	\label{Fig-rr-6} 
\end{figure}
Fig. \ref{Fig-rr-6} shows the average fairness rate versus $K$ for $P=20$dB, $N_{RIS}=20$, $L=2$, $N_{BS}=8$, $M=2$, $n_t=200$, and $\epsilon^c=0.001$.
As can be observed, the fairness rate significantly decreases with $K$. 
 RIS can improve the system performance considerably when $K\leq 4$. However, the benefits of RIS decrease with $K$ and become almost negligible for $K=6$ in this particular example. 
The reason is that, the number of RIS components per user decreases when there are more users in the system. Thus, we have to increase the number of RIS components to compensate for the increment of the number of users. 
Furthermore, we observe that optimizing over the reflecting coefficients is more important when the number of users increases. Indeed, RIS with random reflecting coefficients may even perform much worse than the systems without RIS when $K$ grows.  
Finally, we also observe that the relative mismatch between the Shannon rate and the FBL rate increases with $K$. The mismatch is around $3\%$ for $K=2$, $7\%$ for $K=4$, and $12\%$ for $K=6$.
As indicated before, it happens since the effective SINR decreases with $K$, which in turn yields the further decrements in the FBL rate. 
Note that the gap increase if we reduce $n_t$ and/or $\epsilon^c$ as discussed in the following.

\subsubsection{Impact of packet lengths}
\begin{figure}
    \centering
\includegraphics[width=.4\textwidth]{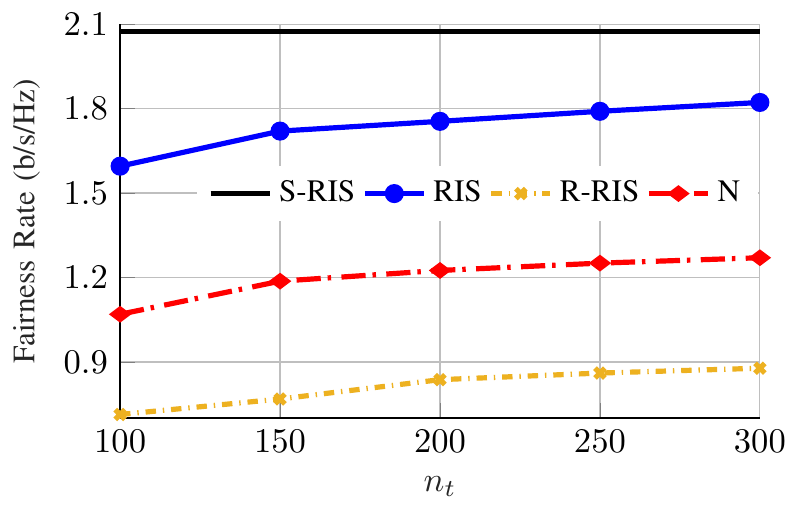}
        \caption{The average fairness rate versus $n_t$ for $\epsilon=10^{-3}$, $N_{RIS}=20$, $L=2$, $K=4$, $M=2$, $P=100$, and $N_{BS}=4$.}
	\label{Fig-rr7}  
\end{figure}
Fig. \ref{Fig-rr7} shows the average fairness rate versus $n_t$ for $N_{BS}=4$, $N_{RIS}=20$, $L=2$, $K=4$, $M=2$, $P=100$, and $\epsilon^c=0.001$. 
As can be observed, by increasing the packet length,  
the gap between the Shannon rate and achievable rate by \eqref{rate} decreases.  
Of course, $n_t$ is not the only important parameter, and there are other effective parameters such as  $\epsilon^c$ and SINR that can have a high impact on the performance and accuracy of the normal approximation in \eqref{rate} as discussed before.

\subsubsection{Impact of $\epsilon^c$}
\begin{figure}[t!]
    \centering
\includegraphics[width=.45\textwidth]{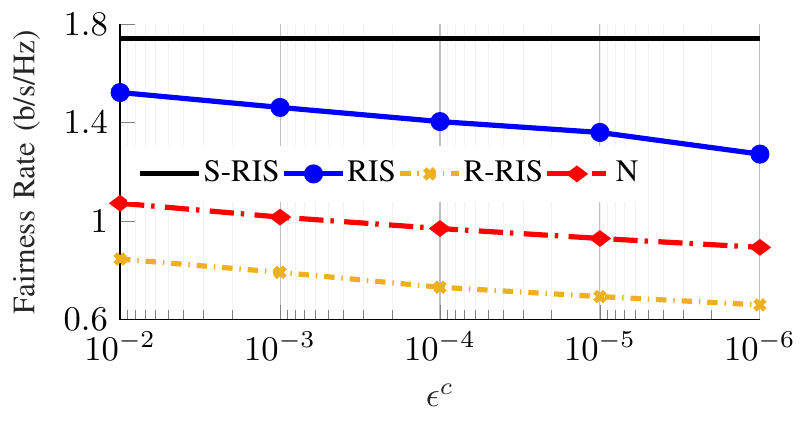}
        \caption{The average fairness rate versus $\epsilon$ for $n_t=200$, $N_{RIS}=20$, $L=2$, $K=4$, $M=2$, $P=10$, and $N_{BS}=4$.}
	\label{Fig-rr8}  
\end{figure}
Fig. \ref{Fig-rr8} shows the average fairness rate versus $\epsilon^c$ for $n_t=200$, $N_{RIS}=20$, $L=2$, $K=4$, $M=2$, and $N_{BS}=4$. 
As can be observed, the average rate increases with $\epsilon^c$. In other words, the lower the decoding error probability is, the lower the average rate is. 
Thus, if we work on ultra-reliable regime with FBL, we have to tolerate some performance loss in the data rate. The more reliable the communication is, the less data rate we can achieve. 
We can also observe in Fig. \ref{Fig-rr8} that RIS can significantly increase the average fairness rate of the system and consequently, for a given date rate, it can highly improve the reliability of the communication link.

\subsubsection{STAR-RIS}
\begin{figure}[t!]
    \centering
\includegraphics[width=.45\textwidth]{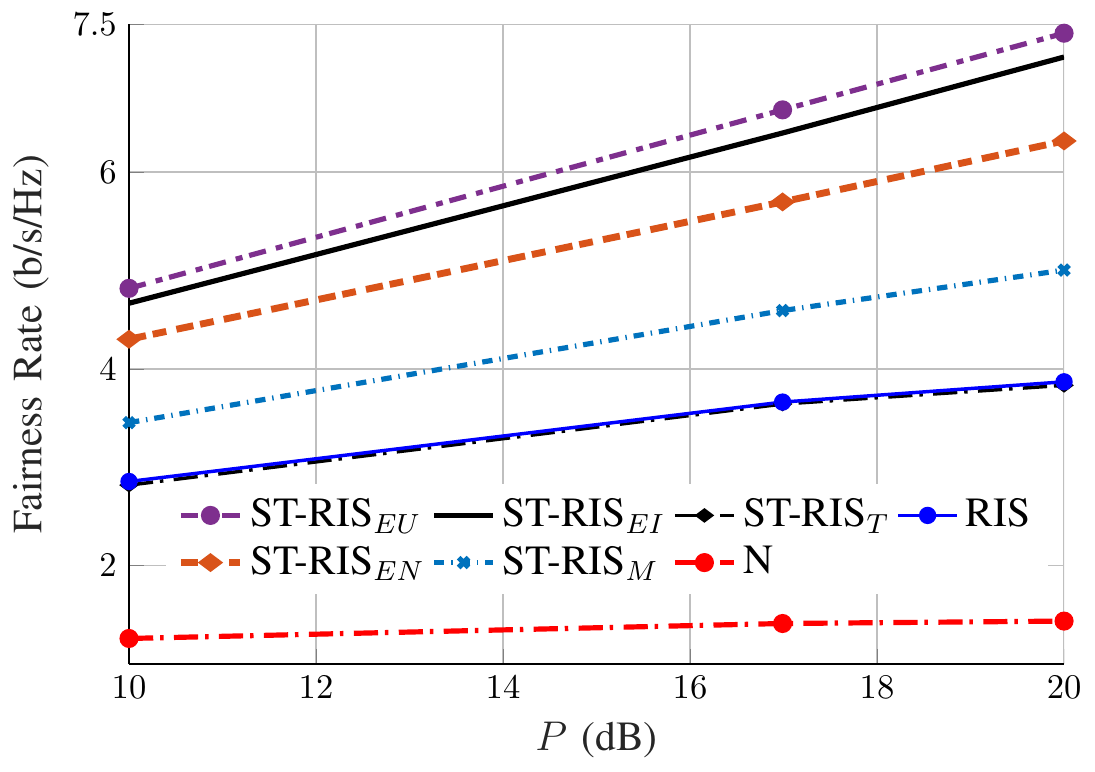}
    \caption{The average fairness rate versus $P$ for $N_{BS}=6$, $N_{RIS}=60$, $L=1$, $K=6$, $M=1$, $n_t=200$, and $\epsilon^c=0.001$.}
	\label{Fig-star} 
\end{figure}
Now we compare the performance of a regular RIS with a STAR-RIS. To this end, we consider a single-cell BC with $K$ users. We assume that the regular RIS can assist only a half of users. In other words, only $K/2$ users are covered by the regular RIS, and the other $K/2$ users do not receive a signal from the regular RIS. 
However, the STAR-RIS can assist all users since it provides a $360^\circ$ coverage. We assume that $K/2$ users are in the reflection space of the STAR-RIS, while the other $K/2$ users are in the transmission space of the STAR-RIS. 
We assume that both the regular and STAR-RIS have the same number of components, i.e., $N_{RIS}$. 
We consider three different strategies for the STAR-RIS. First, half of the RIS components operate only in the reflection mode, and the other half operate only in the transmission mode. We refer to this scheme as the mode switching, similar to \cite{mu2021simultaneously}. 
Second, we assume that all RIS components operate simultaneously in the transmission and reflection modes, which is referred to as the energy splitting mode. Third, we divide each time slot into two sub-slots and assume that all RIS components operate in reflection mode in the first sub-slot, while they all operate in transmission mode in the next sub-slot, which is referred to as the time switching mode.

Fig. \ref{Fig-star} shows the average fairness rate versus      $P$ for $N_{BS}=6$, $N_{RIS}=60$, $L=1$, $K=6$, $M=1$, $n_t=200$, and $\epsilon^c=0.001$. 
As can be observed,  the regular RIS can highly improve the system performance even though it assists only a half of users. The average fairness rate for systems without RIS slightly increases with power budget; however, the fairness rate of RIS-assisted systems (either STAR or regular) almost linearly increases with the power budget. The reason is that the system is interference-limited, and the power budget increment is not necessarily equivalent to SINR improvement. Since the system is not highly overloaded, RIS can manage part of interference, which significantly improves the effective SINR. Additionally, the users are in the cell edge, which implies that they have a relatively weak direct link. Thus, RIS can considerably improve the channel gain, which in turn provide a significant gain, especially when the power budget is high. 

We also observe that the STAR-RIS can outperform the regular RIS with both the MS and ES schemes. However, the STAR-RIS with TS cannot provide any benefit over regular RIS in this particular example since the TS mode covers only reflection or transmission spaces at each time sub-slot, which implies that the TS mode can  assists only a half of users in each sub-slot, similar to the regular RIS. Indeed, the TS mode cannot provide a $360^\circ$ coverage simultaneously, which restricts its applicability in this particular scenario. Moreover, we observe that the STAR-RIS with the ES scheme and different feasibility sets outperforms the MS scheme with the feasibility set $\mathcal{T}_{SI}$. The reason is that the MS scheme can be seen as a lower bound for the performance of STAR-RIS with the ES scheme since the ES scheme encompasses the MS scheme. 
Finally, we observe that the ES scheme with the feasibility set $\mathcal{T}_{SI}$ performs very close to the ES scheme with the feasibility set $\mathcal{T}_{SU}$, which can be considered as an upper bound for the system performance.

\subsubsection{Impact of the feasibility sets}
\begin{figure}[t!]
    \centering
\includegraphics[width=.45\textwidth]{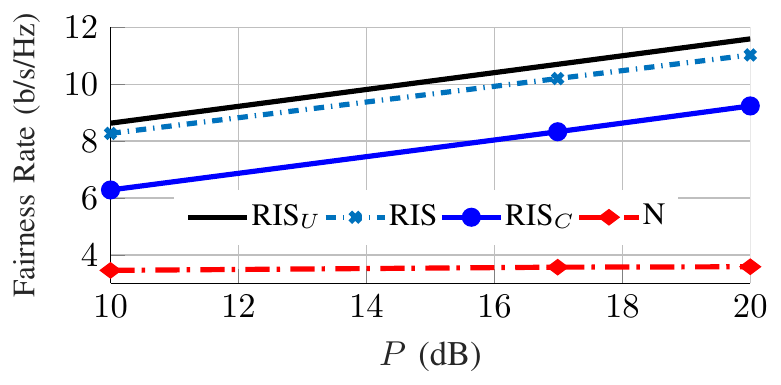}
    \caption{The average fairness rate versus $P$ for $N_{BS}=8$, $N_{RIS}=40$, $L=2$, $K=2$, $M=2$, $n_t=256$, and $\epsilon^c=10^{-5}$.}
	\label{Fig-fs} 
\end{figure}
Fig. \ref{Fig-fs} shows the average fairness rate versus $P$ for $N_{BS}=8$, $N_{RIS}=40$, $L=2$, $K=2$, $M=2$, $n_t=256$, and $\epsilon^c=10^{-5}$. As expected, the feasibility set $\mathcal{T}_U$ outperforms the other feasibility sets. However, the feasibility set $\mathcal{T}_I$ performs very close to the upper bound performance of a passive RIS, which is given by $\mathcal{T}_U$. 
Note that our proposed scheme for the feasibility set $\mathcal{T}_U$ converges to a stationary point of the original problem. Thus, the gap between the upper bound with $\mathcal{T}_U$ and our proposed scheme with $\mathcal{T}_I$ can be seen as an upper bound for the mismatch between our proposed algorithm and a scheme, which attains a stationary point of the original problem. 

\subsubsection{Convergence behavior}
\begin{figure}[t!]
    \centering
\includegraphics[width=.45\textwidth]{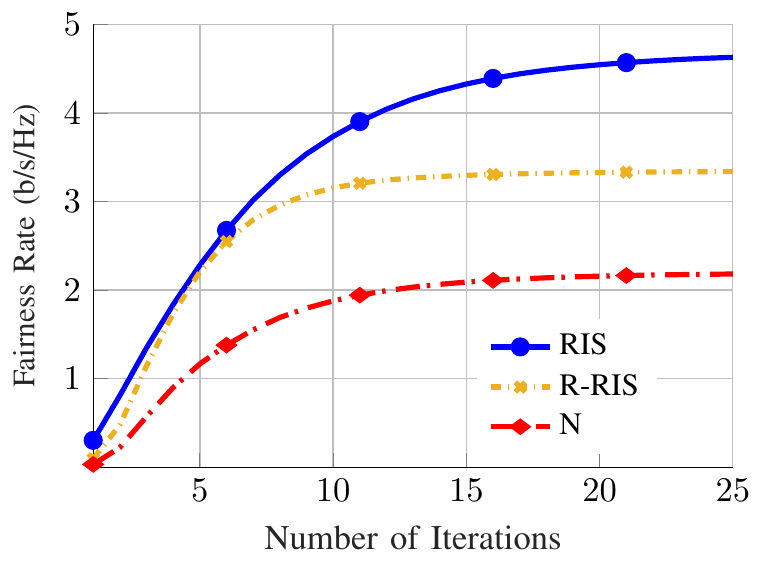}
    \caption{The average fairness rate versus the number of iterations for $P=20$dB, $N_{BS}=8$, $N_{RIS}=40$, $L=2$, $K=2$, $M=2$, $n_t=256$, and $\epsilon^c=10^{-5}$.}
	\label{Fig-fs-rit} 
\end{figure} 
Fig. \ref{Fig-fs-rit} shows the average fairness rate versus  the number of iterations for $P=20$dB, $N_{BS}=8$, $N_{RIS}=40$, $L=2$, $K=2$, $M=2$, $n_t=256$, and $\epsilon^c=10^{-5}$.
As can be observed, the proposed scheme for RIS-assisted systems outperforms the final solution  of the other schemes after a few (less than 10 in this example) iterations. 
Moreover, the schemes converge in about $25$ iterations. 
This figure shows a trade-off between optimality and complexity. Indeed, if there is a very strict latency constraint, the algorithms can be stopped before their convergence when a desired performance has been achieved. 

\subsection{Weighted-sum rate maximization}
\begin{figure}[t!]
    \centering
\includegraphics[width=.4\textwidth]{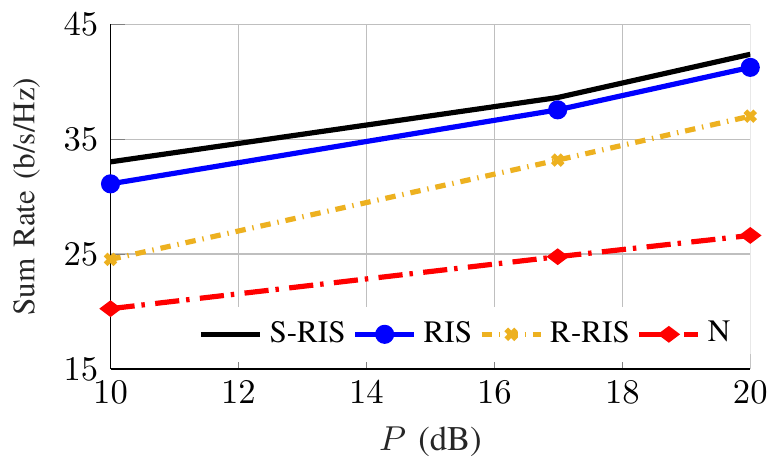}
    \caption{The average sum rate versus $P$ for $N_{BS}=10$, $N_{RIS}=20$, $L=2$, $K=2$, $M=2$, $n_t=200$, and $\epsilon^c=0.001$.}
	\label{Fig-sr} 
\end{figure}
In the previous subsections, we consider the impact of different parameters in the system performance. It can be expected the we observe a similar behavior if we change the objective function. Thus, in this subsection, we provide only one numerical example.  
Fig. \ref{Fig-sr} shows the average sum rate versus $P$ for $N_{BS}=10$, $N_{RIS}=20$, $L=2$, $K=2$, $M=2$, $n_t=200$, and $\epsilon^c=0.001$.
As can be observed, RIS can significantly increase the sum rate of the system. Even an RIS with random components can highly improve the system performance. 
We also observe that there is a small gap between the rate with FBL and the Shannon rate, and the gap decreases with the power budget. 
The reason is that, the mismatch between the rate in \eqref{rate} and the Shannon rate decreases with SINR, and the increment in power budget may result in SINR enhancement since the system is not overloaded.     
Note that the gap is expected to increase if we employ a shorter packet length or operate in a lower decoding error probability. 

\subsection{Global energy efficiency maximization}
\begin{figure}[t!]
    \centering
\includegraphics[width=.4\textwidth]{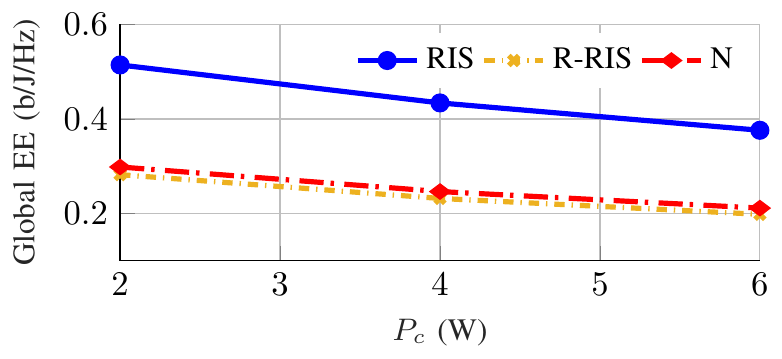}
    \caption{The average GEE versus $p_c$ for $N_{BS}=4$, $N_{RIS}=20$, $L=2$, $K=2$, $M=2$, $n_t=200$, and $\epsilon^c=0.001$.}
	\label{Fig-gee} 
\end{figure}
Fig. \ref{Fig-gee} shows the average GEE versus $p_c$ for $N_{BS}=4$, $N_{RIS}=20$, $L=2$, $K=2$, $M=2$, $n_t=200$, and $\epsilon^c=0.001$. 
In this figure, we assume that the power consumption of each RIS is $1$ W. Thus, the effective $p_c$ for users in systems without RIS is actually $p_c-1/K$ W. 
Note that $p_c$ is the constant power consumption by the devices and is different from the transmission power.   
We can observe through Fig. \ref{Fig-ee} that RIS can highly increase the average GEE of the system if the RIS components are properly optimized.  
However, if the RIS components are randomly chosen, it may happen that RIS may worsen the system performance in this particular example, which is in line with the results in Fig. \ref{Fig-rr}a.
\subsection{ Minimum-weighted energy efficiency maximization}
\begin{figure}[t!]
    \centering
\includegraphics[width=.4\textwidth]{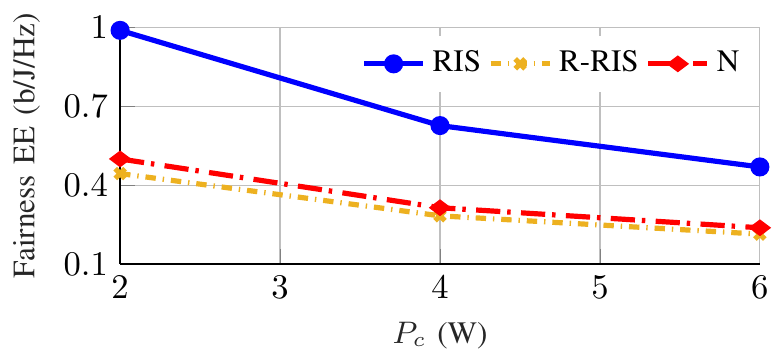}
    \caption{The average fairness EE versus $p_c$ for $N_{BS}=4$, $N_{RIS}=20$, $L=2$, $K=2$, $M=2$, $n_t=200$, and $\epsilon^c=0.001$.}
	\label{Fig-ee} 
\end{figure}
Fig. \ref{Fig-ee} shows the average fairness EE versus $p_c$ for $N_{BS}=4$, $N_{RIS}=20$, $L=2$, $K=2$, $M=2$, $n_t=200$, and $\epsilon^c=0.001$. Similar to Fig. \ref{Fig-gee}, we consider the power consumption of each RIS as $1$ W. As can be observed, RIS can significantly improve the EE of the system if the reflecting coefficients are properly optimized. 
Note that RIS with random reflecting coefficients may not provide any benefits in this particular example, which shows the importance of optimizing the reflecting coefficients.   
We also observe this show in Fig. \ref{Fig-rr}a and Fig. \ref{Fig-gee}.

\begin{figure}
    \centering
\begin{subfigure}{0.24\textwidth}
        \centering
\includegraphics[width=.9\textwidth]{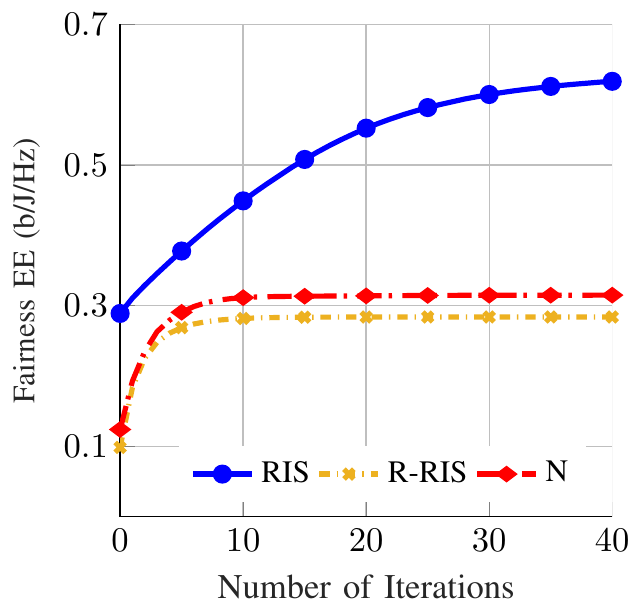}
        \caption{$p_{c}=4W$.}
    \end{subfigure}%
    ~
    \begin{subfigure}{0.24\textwidth}
        \centering
\includegraphics[width=.9\textwidth]{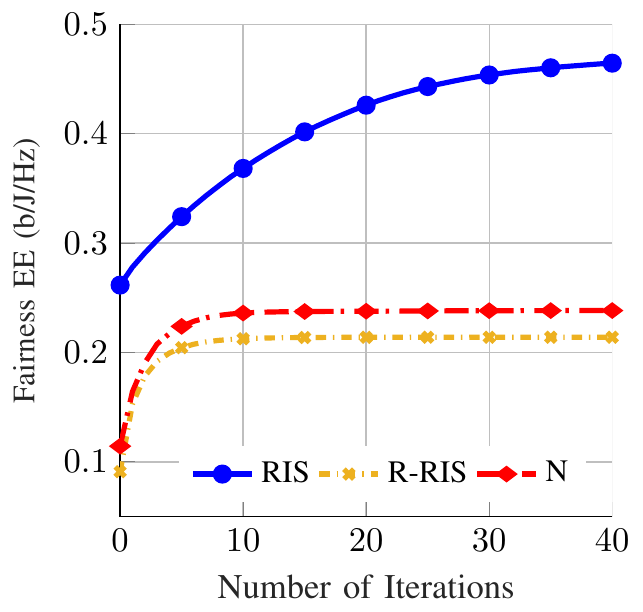}
        \caption{$p_{c}=6W$.}
    \end{subfigure}
    \caption{The average fairness EE versus number of iterations for $N_{BS}=4$, $N_{RIS}=20$, $L=2$, $K=2$, $M=2$, $n_t=200$, and $\epsilon^c=0.001$.}
	\label{Fig-ee2}  
\end{figure}
Fig. \ref{Fig-ee2} shows the average fairness EE versus the number of iterations for $N_{BS}=4$, $N_{RIS}=20$, $L=2$, $K=2$, $M=2$, $n_t=200$, and $\epsilon^c=0.001$. As can be observed, RIS with random components performs worse than the algorithm for systems without RIS. However, when RIS components are properly design, RIS can highly improve the system performance. Additionally, the algorithm for systems without RIS converges with only a few iterations, while the algorithm for RIS-assisted systems require more iterations to converge.  
Note that the initial point of the algorithms is given by the solution of the maximization of the minimum rate to ensure that the initial point is feasible. 

\begin{figure}[t!]
    \centering
\begin{subfigure}{0.24\textwidth}
        \centering
\includegraphics[width=\textwidth]{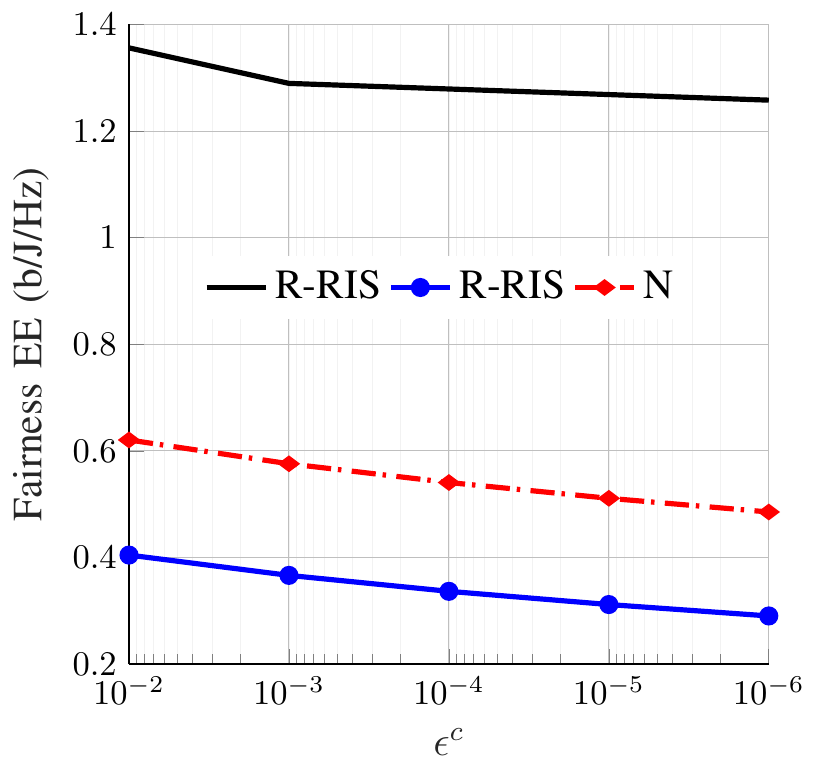}
        \caption{Fairness EE versus $\epsilon^c$ for $n_t=200$ bits and $t_c=0.1$ ms.}
    \end{subfigure}%
    ~
    \begin{subfigure}{0.24\textwidth}
        \centering
\includegraphics[width=\textwidth]{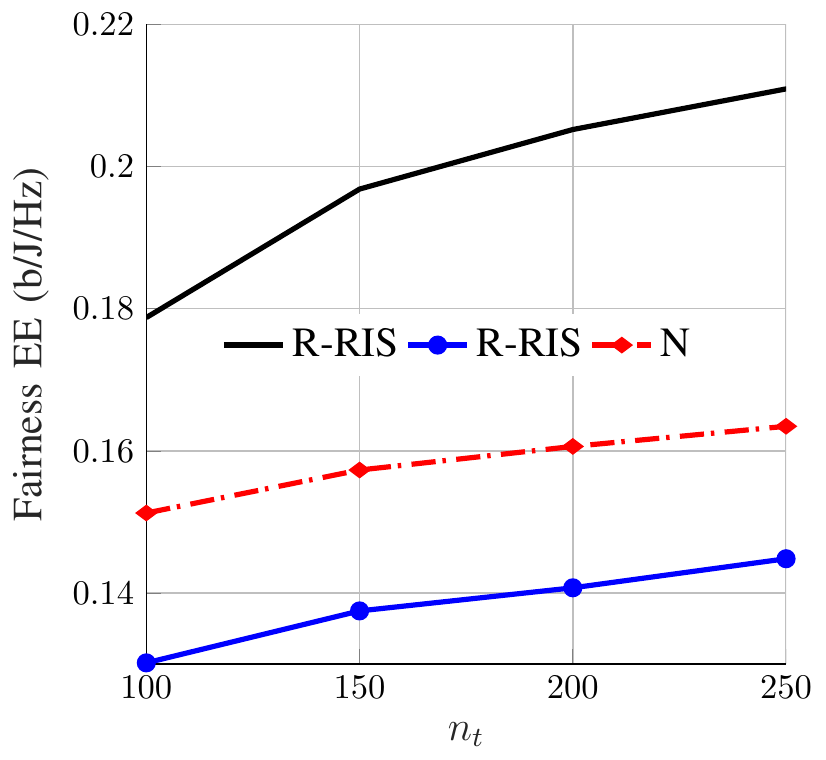}
        \caption{Fairness EE versus $n_t$ for $\epsilon^c=10^{-5}$ and $t_c=5$ ms.}
    \end{subfigure}
    \caption{The average fairness EE versus $\epsilon^c$ and $n_t$ for $N_{BS}=4$, $N_{RIS}=20$, $L=2$, $K=2$, $M=2$. 
}
	\label{Fig-ee3}  
\end{figure}
Fig. \ref{Fig-ee3} shows the average fairness EE versus $\epsilon^c$ and $n_t$ for $N_{BS}=4$, $N_{RIS}=20$, $L=2$, $K=2$, $M=2$. As can be observed, RIS can provide a significant gain for a wide range of variables if the RIS components are properly optimized. As we also observe in Section \ref{sec5a} with the average fairness rate as the performance metric, the system performance is improved by increasing $n_t$ and $\epsilon^c$. Indeed,  the EE decreases if we employ a shorter packet length or if we want to operate with a more reliable link. 
\subsection{Summary}
We show that RIS can significantly improve the spectral and EE of a multi-cell RIS-assisted BC with FBL  and different parameters. 
Moreover, we show that the FBL rate is very close to the Shannon rate for underloaded systems when packet lengths are higher than $200$ bits and $\epsilon$ is higher than $10^{-3}$. This is also in line with the results in \cite[Sec. IV.C]{polyanskiy2010channel}. Note that in underloaded systems, the number of BS antennas is higher than the number of users. 
Thus, the effective SINR of users is expected to be higher in underloaded systems, comparing to overloaded systems, which decreases the gap between the FBL rate and the Shannon rate \cite{polyanskiy2010channel}. 
In particular, our main findings can be summarized as follows:
\begin{itemize} 
\item For a fixed $N_{RIS}$, the benefits of RIS is decreasing in $K$ since the number of RIS components per user decreases with $K$. To compensate for that, we have to increase the number RIS components.

\item The gap between the Shannon rate and the rate by FBL decreases with SINR. It means that the gap increases with $K$ and decreases with $N_{BS}$ and $P$. 

\item RIS can significantly increase the data rate and EE of the system. Equivalently, for a target data rate, RIS can highly improve the reliability of the system by decreasing the error probability. 

\end{itemize}

\section{Conclusion and future work}\label{con-sec}
In this paper, we proposed a general optimization framework to solve a variety of optimization problems in URLLC RIS-assisted networks with FBL. This framework can solve any optimization problem in which the objective and/or constraints are linear functions of the rates/EE of users. Additionally, the framework can be applied to any interference-limited system with TIN. We considered a multi-cell STAR-RIS-assisted BC, as an illustrative example, and specialized the framework to solve the minimum-weighted-rate, weighted-sum-rate, global-EE, and minimum-weighted-EE maximization problems. 
We made realistic assumptions regarding RIS (either regular or STAR-RIS). 
We showed that RIS can considerably improve the spectral and EE of a multi-cell MISO RIS-assisted BC with FBL if RIS components are properly optimized. 
Furthermore, we showed that we may operate close to the Shannon rate with a relatively short packet if the decoding error probability is not very low. Finally, we showed that a STAR-RIS can outperform a regular RIS if the regular RIS cannot cover all the users. 

As future work, it can be interesting to investigate the performance of RIS in URLLC systems with small signaling overheads. Moreover, the global optimal solution of the considered problems is not known, which can be another challenging line for future studies. 
It is also interesting to compare the performance of our techniques with machine-learning-based and/or other data-driven techniques.
\section{Acknowledgment}
The authors would like to thank Dr. Tobias Koch for his valuable comments and discussions during the preparation of this work.

\appendices

\section{Initial points for the optimization framework}\label{simp-sec}
In this appendix, we propose a scheme to heuristically obtain suitable initial points for the proposed optimization framework. 
As indicated in Section \ref{sec=re}, the rate $r_{l,k}$ is negative for $0\leq\gamma_{l,k}\leq\gamma^{(0)}$, which means that $\gamma_{l,k}\leq\gamma^{(0)}$ is not a feasible point. Additionally,  $r_{l,k}$ is strictly increasing in $\gamma_{l,k}$ for all feasible practical points, i.e.,  $\gamma_{l,k}\geq\gamma^{(0)}$. 
To avoid infeasible initial point and get a better performance, we set the solution of the maximization of the minimum SINR as an initial point for our framework. 
The maximization of the minimum SINR can be formulated as
\begin{align}\label{opt-sinr}
\max_{\{\mathbf{x}\}\in\mathcal{X},\{\bm{\Theta}\}\in\mathcal{T},\gamma}&\gamma&
\text{s.t.}\,\,&\gamma_{lk}\geq \gamma&\forall l,k.
\end{align}
Note that for the systems without RIS, we need to optimize only  over $\{\mathbf{x}\}$. However, in this appendix, we consider the most general case. The optimization problem \eqref{opt-sinr} is not convex, but we can obtain a suboptimal solution of \eqref{opt-sinr} by MM, AO and the GDA since it falls into fractional programming and each fraction has a quadratic numerator and denominator \cite{zappone2015energy, soleymani2019improper}.  
That is, we first optimize over the beamforming vectors $\{\mathbf{x}\}$ and obtain $\{\mathbf{x}^{(t)}\}$ while the reflecting coefficients $\{\bm{\Theta}\}$ are fixed to $\{\bm{\Theta}^{(t-1)}\}$. 
Then we alternate and optimize over $\{\bm{\Theta}\}$ while $\{\mathbf{x}\}$ is fixed to $\{\mathbf{x}^{(t)}\}$.
Due to a space restriction, we do not provide the optimization over $\{\mathbf{x}\}$ since the problem has a structure similar to optimizing over $\{\bm{\Theta}\}$. 
 For a given $\{\mathbf{x}^{(t)}\}$, the problem \eqref{opt-sinr} can be written as 
\begin{align}\label{opt-sinr-2}
\max_{\{\bm{\Theta}\}\in\mathcal{T},\gamma}&\gamma&
\text{s.t.}\,\,&\frac{|\mathbf{h}_{lk,l}
\mathbf{x}_{lk}|^2}
{\sigma^2
+
\sum_{\forall [ij]\neq [lk]}|\mathbf{h}_{lk,i}
\mathbf{x}_{ij}|^2}\geq \gamma&\forall l,k,
\end{align}
which is a multiple-ratio FP problem in which both the numerator and denominator are quadratic functions of the channels. 
To solve \eqref{opt-sinr-2}, 
we first employ CCP to approximate  the numerator of $\gamma_{lk}$ for all $l,k$ with a linear lower bound since it is a convex function. 
That is
 \begin{multline}\label{eq-50}
|\mathbf{h}_{lk,l}
\mathbf{x}_{lk}|^2
\geq\hat{s}_{lk}\left(\mathbf{h}_{lk,i}\right)\triangleq
|\mathbf{h}_{lk,l}^{(t-1)}
\mathbf{x}_{lk}^{(t)}|^2
\\
+2\mathfrak{R}\left\{
\mathbf{h}_{lk,l}^{(t-1)}
\mathbf{x}_{lk}^{(t)}
\left(\mathbf{h}_{lk,l}
\mathbf{x}_{lk}^{(t)}-
\mathbf{h}_{lk,l}^{(t-1)}
\mathbf{x}_{lk}^{(t)}\right)
\right\}.
\end{multline}
Substituting \eqref{eq-50} in \eqref{opt-sinr-2}, we  have the following surrogate optimization problem
\begin{align}\label{opt-sinr-2-c}
\max_{\{\bm{\Theta}\}\in\mathcal{T},\gamma}&\gamma&
\text{s.t.}\,\,&\frac{\hat{s}_{lk}\left(\mathbf{h}_{lk,i}\right)
}
{\sigma^2
+
\sum_{\forall [ij]\neq [lk]}|\mathbf{h}_{lk,i}
\mathbf{x}_{ij}|^2}\geq \gamma&\forall l,k,
 \end{align}
which is non-convex, but can be solved by GDA. 
That is, we iteratively solve 
\begin{subequations}\label{opt-sinr-3-c}
\begin{align}
\max_{\{\bm{\Theta}\}\in\mathcal{T},\gamma}\!&\gamma\!\\
\text{s.t.}\,\,&
\hat{s}_{lk}\left(\cdot\right)\!
-\mu^{(t,m)}
\!
\left(
\!\!\sigma^2
\!\!+\!\!\!\!\!\!
\sum_{\forall [ij]\neq [lk]}
\!\!\!\!
|\mathbf{h}_{lk,i}
\mathbf{x}_{ij}|^2
\!\!
\right)\!
\geq \gamma,\forall l,k,
 \end{align}
\end{subequations}
and update $\mu^{(t,m)}$ as
\begin{equation*}
\mu^{(t,m)}=\min_{l,k}\left(\frac{\hat{s}_{lk}\left(\mathbf{h}_{lk,i}^{(t,m)}\right)
}
{\sigma^2
+
\sum_{\forall [ij]\neq [lk]}|\mathbf{h}_{lk,i}^{(t,m)}
\mathbf{x}_{ij}|^2}\right),
\end{equation*}
where $\mathbf{h}_{lk,i}^{(t,m)}$ is the initial point at the $m$-th iteration of \eqref{opt-sinr-3-c}, which is the solution of the previous step. 
The optimization problem \eqref{opt-sinr-3-c} is convex when $\mathcal{T}$ is a convex set, i.e., when we consider $\mathcal{T}_U$ for regular RIS and $\mathcal{T}_{SU}$ for STAR-RIS. 
We can convexify $\mathcal{T}_I$, $\mathcal{T}_C$, $\mathcal{T}_{SI}$, and $\mathcal{T}_{SN}$ similar to Section \ref{sec-opt-ref}. 
Since it is straightforward to do so, we do not repeat them here.

\section{Useful inequalities}

In this appendix, we provide some inequalities, which are widely used in this work. 
Consider real and positive variables 
$x$, $\bar{x}$, $y$ and $\bar{y}$. Then, the following inequality holds for all $x,y,\bar{x},\bar{y}$ \cite[Eq. (75)]{nasir2020resource}
\begin{equation}
\sqrt{xy}\leq \frac{\sqrt{\bar{x}}}{2\sqrt{\bar{y}}}y+\frac{\sqrt{\bar{y}}}{2\sqrt{\bar{x}}}x,
\end{equation}
For the case that $y=\bar{y}=1$, we have
\begin{equation}\label{eq-11}
\sqrt{x}\leq \frac{\sqrt{\bar{x}}}{2}+\frac{x}{2\sqrt{\bar{x}}}.
\end{equation}
Additionally, the following inequality holds for all feasible $x,y,\bar{x},\bar{y}$ \cite[Eq. (76)]{nasir2020resource}
\begin{equation}\label{eq-12}
\frac{x^2}{y}\geq \frac{\bar{x}}{\bar{y}}\left(2x-\frac{\bar{x}}{\bar{y}}y\right).
\end{equation}
When $x$ and $\bar{x}$ are complex, \eqref{eq-12} is modified to 
\begin{equation}\label{eq-12-mo}
\frac{|x| ^2}{y}\geq \frac{2\mathfrak{R}\{\bar{x}^*x\}}{\bar{y}}-\frac{|\bar{x}|^2}{\bar{y}^2}y.
\end{equation}

Now, consider positive real-valued variables $y$ and $\bar{y}$, and complex-valued variables $x$ and $\bar{x}$. Then, the following inequality holds for all feasible   $y,\bar{y},x,\bar{x}$ \cite[Lemma 2]{soleymani2022improper}
\begin{multline}
\label{eq-13}
\ln\left(1+\frac{|x|^2}{y}\right)\geq \ln\left(1+\frac{|\bar{x}|^2}{\bar{y}}\right)
-\frac{|\bar{x}|^2}{\bar{y}}
\\
+\frac{2\mathfrak{R}\{\bar{x}^*x\}}{\bar{y}}
-\frac{|\bar{x}|^2}{\bar{y}}\frac{|x|^2+y}{|\bar{x}|^2+\bar{y}}.
\end{multline}

\section{Proof of Lemma 3}
\label{app=b}
The rate of users consists of two parts: the Shannon rate and the term related to the channel dispersion. We first obtain a concave lower-bound for the Shannon rate. To this end, we employ the inequality in \eqref{eq-13}, which gives us
\begin{align}
\nonumber
{r}_{s,lk}\geq\hat{r}_{s,lk}^{(t-1)}&=\log\left(1+\gamma_{lk}^{(t-1)}\right)
-
\gamma_{lk}^{(t-1)}
\\
\nonumber
&
\hspace{.3cm}
+
\frac{2\mathfrak{R}
\left\{
\left(
\mathbf{h}_{lk,l}
\mathbf{x}_{lk}^{(t-1)}
\right)^*
\mathbf{h}_{lk,l}
\mathbf{x}_{lk}
\right\}
}{
\sigma^2
+
\sum_{ij}|\mathbf{h}_{lk,i}
\mathbf{x}_{ij}^{(t-1)}|^2-|\mathbf{h}_{lk,l}
\mathbf{x}_{lk}^{(t-1)}|^2
}
\\
&
\hspace{.3cm}
-
\gamma_{lk}^{(t)}
\frac{
\sigma^2
+
\sum_{ij}|\mathbf{h}_{lk,i}
\mathbf{x}_{ij}|^2
}{
\sigma^2
+
\sum_{ij}|\mathbf{h}_{lk,i}
\mathbf{x}_{ij}^{(t-1)}|^2
}.
\label{r-hat}
\end{align}
Now, we obtain a concave lower-bound for $-\delta_{lk}$, which is equivalent to finding a convex upper bound for $\sqrt{V_{lk}}$. Applying the inequality \eqref{eq-11}, we have 
\begin{align*}
\sqrt{V_{lk}}&\leq 
\frac{\sqrt{V_{lk}^{(t)}}}{2}
+
\frac{\gamma_{lk}}{\sqrt{V_{lk}^{(t)}}\left(1+\gamma_{lk}\right)}.
\end{align*}
Replacing $\gamma_{lk}$ by \eqref{sinr}, we have
\begin{multline}\label{eq19}
\delta_{lk}=\frac{Q^{-1}(\epsilon^c)}{\sqrt{n_t}}\sqrt{V_{lk}}\leq \tilde{\delta}_{lk}=
\frac{Q^{-1}(\epsilon^c)\sqrt{V_{lk}^{(t)}}}{2\sqrt{n_t}}
\\
+\frac{Q^{-1}(\epsilon^c)}{\sqrt{n_tV_{lk}^{(t)}}}
\left(1-
\underbrace{
\frac{
\sigma^2+\sum_{[ij]\neq [lk]}|\mathbf{h}_{lk,i}
\mathbf{x}_{ij}|^2
}{\sigma^2+\sum_{ij}|\mathbf{h}_{lk,i}
\mathbf{x}_{ij}|^2}
}_{\zeta_{lk}}
\right).
\end{multline}
Unfortunately, $\tilde{\delta}_{lk}$ is not a convex function since $\zeta_{lk}$ is not concave. 
However, we can apply \eqref{eq-12-mo} to obtain a concave lower bound for $\zeta_{lk}$ (or equivalently a convex upper bound for $\tilde{\delta}_{lk}$) as
\begin{multline}\label{eq20}
\zeta_{lk}\geq
2\frac{
\sigma^2+\sum_{[ij]\neq [lk]}\mathfrak{R}\left\{\left(\mathbf{h}_{lk,i}
\mathbf{x}_{ij}^{(t-1)}\right)^*
\mathbf{h}_{lk,i}
\mathbf{x}_{ij}
\right\}
}{\sigma^2+\sum_{ij}|\mathbf{h}_{lk,i}
\mathbf{x}_{ij}^{(t-1)}|^2}
\\
-
\zeta_{lk}^{(t-1)}\frac{
\sigma^2+\sum_{ij}|\mathbf{h}_{lk,i}
\mathbf{x}_{ij}|^2
}{\sigma^2+\sum_{ij}|\mathbf{h}_{lk,i}
\mathbf{x}_{ij}^{(t-1)}|^2}.
\end{multline}
Substituting \eqref{eq20} into \eqref{eq19}, we have 
\begin{multline}\label{eq21}
\delta_{lk}\leq \tilde{\delta}_{lk}\leq \hat{\delta}_{lk}=
\frac{Q^{-1}(\epsilon^c)\sqrt{V_{lk}^{(t)}}}{2\sqrt{n_t}}
+\frac{Q^{-1}(\epsilon^c)}{\sqrt{n_tV_{lk}^{(t)}}}
\\
-\frac{2Q^{-1}(\epsilon^c)}{\sqrt{n_tV_{lk}^{(t)}}}
\frac{
\sigma^2+\sum_{[ij]\neq [lk]}\mathfrak{R}\left\{\left(\mathbf{h}_{lk,i}
\mathbf{x}_{ij}^{(t-1)}\right)^*
\mathbf{h}_{lk,i}
\mathbf{x}_{ij}
\right\}
}{\sigma^2+\sum_{ij}|\mathbf{h}_{lk,i}
\mathbf{x}_{ij}^{(t-1)}|^2}
\\
+
\frac{Q^{-1}(\epsilon^c)\zeta_{lk}^{(t-1)}}{\sqrt{n_tV_{lk}^{(t)}}}
\frac{
\sigma^2+\sum_{ij}|\mathbf{h}_{lk,i}
\mathbf{x}_{ij}|^2
}{\sigma^2+\sum_{ij}|\mathbf{h}_{lk,i}
\mathbf{x}_{ij}^{(t-1)}|^2}
.
\end{multline}
Finally, the concave lower bound for $r_{lk}$ is
\begin{equation}
r_{lk}\geq \tilde{r}_{lk}^{(t-1)}=
\hat{r}_{s,lk}^{(t-1)}
-
\hat{\delta}_{lk}^{(t-1)}.
\end{equation}
If we substitute the values for $\hat{r}_{s,lk}^{(t-1)}$ and $\hat{\delta}_{lk}^{(t-1)}$, and simplify the equations, we can easily obtain \eqref{eq=lem=q}.

\bibliographystyle{IEEEtran}
\bibliography{ref2}
\end{document}